\documentclass{article}

\usepackage{fullpage} %
\usepackage{xspace} %
\usepackage{xcolor} %

\usepackage{graphicx} %
\usepackage{amsmath} %
\usepackage{amssymb}
\usepackage{paralist}%
\usepackage{hyperref}
\usepackage{physics}
\usepackage{qcircuit}
\usepackage{braket}
\usepackage{bm}
\usepackage{xifthen}
\usepackage{caption}
\usepackage{subcaption}
\usepackage{mathtools}
\usepackage{algorithm}
\usepackage[noend]{algpseudocode}

\usepackage{authblk}

\usepackage{multirow}
\usepackage{makecell}
\usepackage{lscape}
\usepackage{tikz, calc}
\usepackage{wrapfig}
\usepackage{fontawesome}
\usepackage{amsthm}
\usetikzlibrary{calc}

\DeclarePairedDelimiter{\ceil}{\lceil}{\rceil}
\DeclarePairedDelimiter{\floor}{\lfloor}{\rfloor}

\theoremstyle{plain}
\newtheorem{theorem}{Theorem} 
\newtheorem{lemma}{Lemma}
\theoremstyle{definition}

\newcommand*\R{\mathbb{R}}

\newcommand*\g{\gamma}
\newcommand*\bt{\beta}

\renewcommand\H{\mathcal{H}}

\newcommand\K[1]{\ket{#1}}

\newcommand\bmat[1]{\begin{bmatrix} #1 \end{bmatrix}}

\newcommand\cs{\text{c}}
\newcommand\sn{\text{s}}
\renewcommand\O{\mathcal{O}}

\newcommand{\E}{\mathop{\mathbb{E}}}

\newcommand{\algprobm}[1]{{\sc #1}\xspace}
\newcommand{\MC}{\algprobm{MaxCut}}

\newcommand{\LMC}{\algprobm{LocalMaxCut}}
\newcommand{\bit}{\{0,1\}}
\newcommand{\bitn}{\{0,1\}^n}
\newcommand{\bitz}{\{-1,+1\}}

\newcommand{\prob}[1]{\Probability{[#1]}}

\makeatletter
\newcommand{\namelabel}[1]{%
  \phantomsection
  \renewcommand{\@currentlabel}{#1}
  \label{#1}
}
\makeatother

\renewenvironment{abstract}
 {\small
  \begin{center}
  \vspace{1em}
  \large\bfseries\abstractname\vspace{0pt}
  \end{center}
  \list{}{
    \setlength{\leftmargin}{2cm}%
    \setlength{\rightmargin}{\leftmargin}%
  }%
  \item\relax}
 {\endlist}

\begin{document}

\title{A quantum advantage over classical for local max cut}
\author[2]{Charlie Carlson}
\author[1]{Zackary Jorquera}
\author[1]{Alexandra Kolla}
\author[1]{Steven Kordonowy}
\affil[1]{University of California - Santa Cruz}
\affil[2]{University of California - Santa Barbara}
\date{}
\maketitle

\begin{abstract}
  We compare the performance of a quantum local algorithm to a similar classical counterpart on a well-established combinatorial optimization problem \LMC. We show that a popular quantum algorithm first discovered by Farhi, Goldstone, and Gutmannn \cite{FGG} called the quantum optimization approximation algorithm (QAOA) has a computational advantage over comparable local classical techniques on degree-3 graphs. These results hint that even small-scale quantum computation, which is relevant to the current state-of the art quantum hardware, could have significant advantages over comparably simple classical computation. 
\end{abstract}

\section{Introduction}

With the advent of quantum computers rose the desire to explore their potentially tremendous computational power. A powerful line of research questions that has gained traction revolves around proving so-called quantum ``supremacy'' over classical computation. Namely, can quantum computers perform certain important computational tasks much faster than classical computers? Shor's factoring algorithm \cite{shor} was proof that quantum computers can indeed outperform classical computers when it comes to solving questions of great significance. However, Shor and related algorithms require large-scale quantum computers in order to show any advantage whereas today's state of the art quantum hardware is still limited to a few dozen working qubits \cite{google}. Thus a new important line of questions rears its head: can algorithms that are meaningful to run in small to medium scale quantum computers, such as local quantum algorithms, outperform their local classical counterparts? And if so, what type of problems can they solve better or faster than classical local computation? In this paper, we give evidence that for certain local combinatorial optimization problems, local quantum algorithms exhibit computational advantage over comparable classical algorithms. 

Given a graph $G = (V,E)$, a \textit{cut} is an assignment of the vertices to $+$ and $-$ and an edge is cut if its endpoints are assigned $+-$ or $-+$. That is, a cut is a function $\tau : V \rightarrow \{+,-\}$ where an edge $e = (u,v)$ is cut when $\tau(u) \neq \tau(v)$. The (unweighted) \MC problem asks to find a cut that maximizes the number of cut edges. This problem arises naturally when minimizing the energy of anti-ferromagnetic Heisenberg spin systems in which the goal is to assign opposing spins to neighboring nodes \cite{spin-glass}. Finding optimal \MC solutions is computationally intractable so we relax to an easier problem \cite{karp}. A vertex $v$ is \textit{(locally) satisfied} under $\tau$ if at least half of the edges incident to $v$ are cut. A \textit{locally maximal cut} is one in which all vertices are satisfied. Finding a locally maximal cut is not hard to do. The unweighted version can be done in $O(n^2)$ steps in the worst case if one has access to the full graph. Restricted to local computation, this is not so trivial. The \LMC problem is the optimization version in which we want a cut that satisfies as many vertices as possible.

A local graph algorithm is a technique to distribute computation over vertices wherein each individual computation requires only local neighborhood information. For simplicity we assume our graphs are unweighted and regular with degree $d$. In general, local algorithms are approximations of ``global'' algorithms in which we have access to the full graph at all steps of the algorithm. For an optimization problem, let $OPT(G)$ be the optimal value for graph $G$. An algorithm that ouputs value $ALG(G)$ is an \textit{$\alpha$-approximation} if $ALG(G) \geq \alpha \cdot OPT(G)$ for any $G$. The best global algorithm for \MC gives an $0.878$-approximation was discovered by Goemans and Williamson and is optimal under further common complexity assumptions is optimal \cite{goewim, khot,hastad2000}. This algorithm relies on solving semidefinite programs which use global information to solve. Restricting to local computation, the best classical techniques produce a cut satisfying $1/2 + \Omega(1/\sqrt{d})$ fraction of the edges \cite{Shearer1992,barak_etal,threshold} on triangle-free graphs. These algorithms all have similar structure: randomly draw an initial cut and then every vertex queries neighboring vertices to determine how it should update its assignment. The classical algorithm our paper considers proposes an update step that generalizes \cite{threshold}.

In 2014, Farhi et al. introduced the quantum approximation optimization algorithm (QAOA) as a way to solve constraint satisfaction problems \cite{FGG}. The QAOA first encodes the linear objective function into the language of Hermitian operators and then relies on mixing techniques from quantum mechanics to round to a good solution. In particular, they proved lower bounds for the performance of the QAOA on \MC performs similar to the aforementioned classical local techniques, satisfying $1/2 + \Omega\left(\frac1{\log{d}\sqrt{d}} \right)$ percentage of edges \cite{FGG2}. The back-and-forth between the QAOA performance and the classical techniques on specific combinatorial problems culminated with Hastings' \cite{hastings2019classical} description of the local tensor meta-algorithm. Both the QAOA and the local classical algorithms fall within this technique. Moreover, Hastings proves that a classical local tensor algorithm outperform the QAOA on \MC on triangle-free graphs. Hastings local tensor algorithms are typically iterative but this work focuses on single round algorithms.

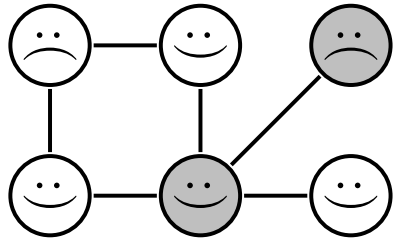
\begin{wrapfigure}{r}{0.5\textwidth}
  \centering
    \begin{tikzpicture}[scale=2, every node/.style={circle, draw, scale=2, line width=1.5pt, minimum size=15pt, inner sep=0pt}]
      \node [rotate = 270] (1) at (0,1) {:(};
      \node [rotate = 270] (2) at (1,1) {:)};
      \node [rotate = 270] (3) at (2,1) [fill=gray!50] {:(};
      \node [rotate = 270] (4) at (0,0) {:)};
      \node [rotate = 270] (5) at (1,0) [fill=gray!50] {:)};
      \node [rotate = 270] (6) at (2,0) {:)};
      
      \draw [line width=1.5pt] (1) -- (2) -- (5) -- (4) -- (1) (3) -- (5) -- (6);
    \end{tikzpicture}
    \caption{A non-locally maximal cut on $G$. Grey nodes correspond to -1 assignments while white nodes correspond to +1. A smiley face correponds to at least half of its edges are cut (i.e., the spin is locally satisfied).}
    \label{fig:cut_graph1}
  \end{wrapfigure}

An intuitive definition of \LMC is to rephrase the definition of ``locally maximal'' to focus on vertices rather than edges: $v \in V$ is \textit{satisfied} when at least half of the edges incident to $v$ are cut. A cut is then locally maximal if and only if all of its vertices are satisfied.\footnote{Indeed, if a vertex were to be unsatisfied, then flipping its assignment guarantees it will be satisfied and so at least one more edge is now cut.} Any true maximum cut is locally maximal and every graph has some maximum cut therefore every graph contains a cut satisfying all vertices. We define \LMC to be the optimization problem of finding a cut that satisfies as many vertices as possible.

There are a few important distinctions between \LMC and \MC. Firstly, \LMC is a valid relaxation to \MC. A true maximum cut is guaranteed to be locally maximal but the converse is not true. See figure \ref{fig:LMC_vs_mc_graphs}. Since every graph must contain some maximum cut, this implies that $OPT(\LMC) = \abs{V}$ no matter the underlying graph. This is different than \MC in which $OPT(\MC)$ could be very far from $\abs{E}$. For example, in the complete graph, $OPT(\MC) = 1/2 + \frac1{d}$. Another important distinction are the problem localities. The \MC objective can be broken into $\abs{E}$ subroutines that act on only 2 input bits per function, namely the endpoints of an edge $uv \in E$. Since each subroutine relies on the assignment of only two spins, \MC is 2-local as a problem and this does not change no matter how complex the underlying interaction graph. On the other hand, \LMC's objective is decomposed into $\abs{V}$ functions that act on $d+1$ input bits (the assignment of the full neighborhood of a vertex) which means that \LMC is $(d+1)$-local. This distinction in locality seems to be crucial. With \MC, we saw that local algorithms can outperform the QAOA. One main idea from this paper is that the QAOA is able to exploit the larger locality of \LMC and thus outperform classial techniques in larger-degree graphs.

The main contribution of our paper is the following two theorems on low-degree graphs.

\begin{theorem}
  For degree-2 graphs, there is a classical algorithm that outputs a cut with $\geq 0.95n$ locally satisfied vertices whereas the QAOA can only guarantee $<0.94n$ satisfied vertices.
\end{theorem}

\begin{theorem}
  For degree-3 graphs with sufficiently large girth, the QAOA satisfies $\geq 0.819n$ vertices while the one-round classical algorithm can only guarantee $< 0.8n$ satisfied vertices in expectation.
\end{theorem}

As these theorems show, for the simplest degree-2 graphs, the classical algorithm outperforms the QAOA. In contrast, graphs with degree-3 allow for enough complexity that the QAOA outperforms the basic probabilistic classical algorithm. These algorithms are both probabilistic and the theorems hold in expectation. In both the classical and quantum algorithms, we provide an explicit formula for the probability that the algorithm satisfies a single vertex and then extend to the expected number of satisfied vertices by linearity of expectation. The classical probabilities are derived from first principles and optimized using calculus techniques and numerical optimization. The QAOA probabilities are much more involved and are calculated using a techniques introduced in \cite{ryananderson2018quantum} along with numerical optimization.

One of the major challenges that the authors faced when trying to explicitly analyze these local algorithms on \LMC is the rate of growth of the system size. In particular, the local Hamiltonian terms grow with the graph degree $d$ which leads to complicated neighborhood interactions. In the \MC case, no matter how large $d$ gets, the edge Hamiltonian stays simple: an edge is satisfied if the endpoints are colored differently. In contrast, in \LMC, the local terms that encode if a vertex is satisfied involve querying $O(d)$ neighboring vertices. This property results in an $O(2^d)$-sized local Hamiltonian term and much more complicated dependency structures between neighborhoods and thus very challenging probability calculations that the authors had to overcome.
In the quantum setting, as a result of independent interest, our techniques also lead to improving the technique introduced in \cite{ryananderson2018quantum} to more general Hamiltonians.


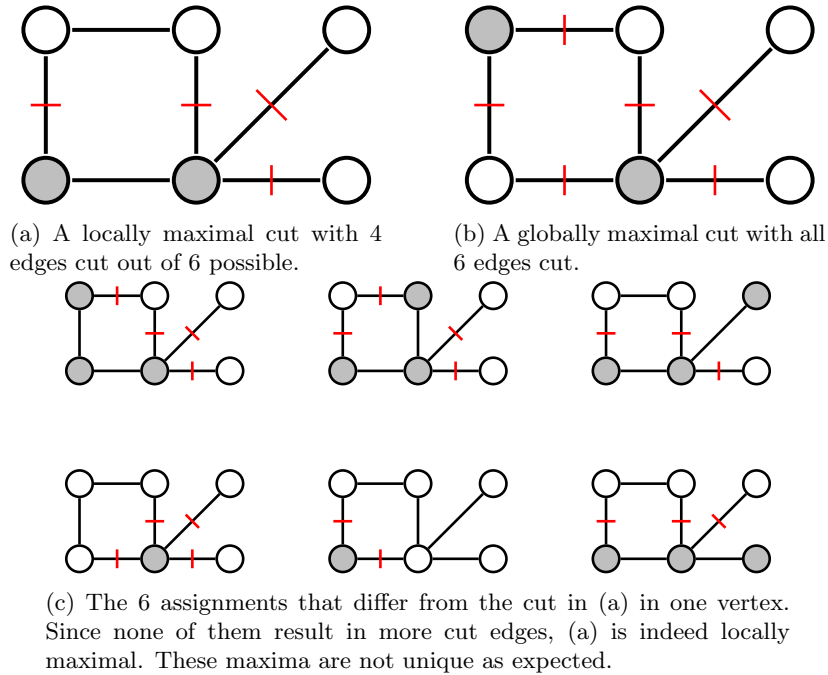
\begin{figure}[ht!]
    \centering
    \begin{subfigure}{0.3\linewidth}
        \centering
        \begin{tikzpicture}[scale=2, every node/.style={circle, draw, scale=2, line width=1.5pt, inner sep=3pt}]
            \node (1) at (0,1) {};
            \node (2) at (1,1) {};
            \node (3) at (2,1) {};
            \node (4) at (0,0) [fill=gray!50] {};
            \node (5) at (1,0) [fill=gray!50] {};
            \node (6) at (2,0) {};
            
            \draw [line width=1.5pt] (1) -- (2) -- (5) -- (4) -- (1) (3) -- (5) -- (6);
            
            \draw [line width=1pt, red] (-0.1, 0.5) -- (.1, 0.5);
            \draw [line width=1pt, red] (1-0.1, 0.5) -- (1+.1, 0.5);
            \draw [line width=1pt, red] (1.5, -.1) -- (1.5, 0.1);
            \draw [line width=1pt, red] (1.6, 0.4) -- (1.4, 0.6);
        \end{tikzpicture}
        \caption{A locally maximal cut with 4 edges cut out of 6 possible.}
        \label{fig:local-max-ex}
    \end{subfigure}\hspace{0.05\linewidth}
    \begin{subfigure}{0.3\linewidth}
        \centering
        \begin{tikzpicture}[scale=2, every node/.style={circle, draw, scale=2, line width=1.5pt, inner sep=3pt}]
            \node  (1) at (0,1) [fill=gray!50] {};
            \node  (2) at (1,1) {};
            \node (3) at (2,1) {};
            \node (4) at (0,0)  {};
            \node (5) at (1,0) [fill=gray!50] {};
            \node (6) at (2,0) {};
            
            \draw [line width=1.5pt] (1) -- (2) -- (5) -- (4) -- (1) (3) -- (5) -- (6);
            
            \draw [line width=1pt, red] (-0.1, 0.5) -- (.1, 0.5);
            \draw [line width=1pt, red] (1-0.1, 0.5) -- (1+.1, 0.5);
            \draw [line width=1pt, red] (1.5, -.1) -- (1.5, 0.1);
            \draw [line width=1pt, red] (0.5, -.1) -- (0.5, 0.1);
            \draw [line width=1pt, red] (0.5, 0.9) -- (0.5, 1.1);
            \draw [line width=1pt, red] (1.6, 0.4) -- (1.4, 0.6);
        \end{tikzpicture}
        \caption{A globally maximal cut with all 6 edges cut.}
        \label{fig:global-max-ex}
    \end{subfigure}
    \begin{subfigure}{0.6\linewidth}
        \centering
        \begin{tikzpicture}[scale=1, every node/.style={circle, draw, scale=1, line width=1pt, minimum size=10pt, inner sep=0pt}, text node/.style={draw=none}]
            \node (11) at (0,3.5) [fill=gray!50] {};
            \node (21) at (3.5,3.5) {};
            \node (31) at (7,3.5) {};
            \node (41) at (0,1) {};
            \node (51) at (3.5,1) {};
            \node (61) at (7,1) {};
            
            \node (12) at ($(11)+(1,0)$) {};
            \node (13) at ($(11)+(2,0)$) {};
            \node (14) at ($(11)+(0,-1)$) [fill=gray!50] {};
            \node (15) at ($(11)+(1,-1)$) [fill=gray!50] {};
            \node (16) at ($(11)+(2,-1)$) {};
            
            \node (22) at ($(21)+(1,0)$) [fill=gray!50] {};
            \node (23) at ($(21)+(2,0)$) {};
            \node (24) at ($(21)+(0,-1)$) [fill=gray!50] {};
            \node (25) at ($(21)+(1,-1)$) [fill=gray!50] {};
            \node (26) at ($(21)+(2,-1)$) {};
            
            \node (32) at ($(31)+(1,0)$) {};
            \node (33) at ($(31)+(2,0)$) [fill=gray!50] {};
            \node (34) at ($(31)+(0,-1)$) [fill=gray!50] {};
            \node (35) at ($(31)+(1,-1)$) [fill=gray!50] {};
            \node (36) at ($(31)+(2,-1)$) {};
            
            \node (42) at ($(41)+(1,0)$) {};
            \node (43) at ($(41)+(2,0)$) {};
            \node (44) at ($(41)+(0,-1)$) {};
            \node (45) at ($(41)+(1,-1)$) [fill=gray!50] {};
            \node (46) at ($(41)+(2,-1)$) {};
            
            \node (52) at ($(51)+(1,0)$) {};
            \node (53) at ($(51)+(2,0)$) {};
            \node (54) at ($(51)+(0,-1)$) [fill=gray!50] {};
            \node (55) at ($(51)+(1,-1)$) {};
            \node (56) at ($(51)+(2,-1)$) {};
            
            \node (62) at ($(61)+(1,0)$) {};
            \node (63) at ($(61)+(2,0)$) {};
            \node (64) at ($(61)+(0,-1)$) [fill=gray!50] {};
            \node (65) at ($(61)+(1,-1)$) [fill=gray!50] {};
            \node (66) at ($(61)+(2,-1)$) [fill=gray!50] {};

            \foreach \i in {1,2,...,6}{
                \draw [line width=1pt] (\i3) -- (\i5) -- (\i2) -- (\i1) -- (\i4) -- (\i5) -- (\i6);
            }
            
            \foreach \x / \y in {0/3.5,1/2.5,3.5/3.5,4.5/2.5,8/2.5,0/0,1/0,3.5/0}{
                \draw [line width=1pt, red] ($(\x,\y)+(0.5, 0)+(0,-1/8)$) -- ($(\x,\y)+(0.5, 0)+(0,1/8)$);
            }
            
            \foreach \x / \y in {1/3.5,3.5/3.5,7/3.5,8/3.5,1/1,3.5/1,7/1,8/1}{
                \draw [line width=1pt, red] ($(\x,\y)+(0, -0.5)+(-1/8,0)$) -- ($(\x,\y)+(0, -0.5)+(1/8,0)$);
            }
            
            \foreach \i in {1,2,4,6}{
                \draw [line width=1pt, red] ($(\i2)+(0.4116, -0.4116)$) -- ($(\i6)+(-0.4116, +0.4116)$); 
            }
        \end{tikzpicture}
        \caption{The 6 assignments that differ from the cut in (a) in one vertex. Since none of them result in more cut edges, (a) is indeed locally maximal. These maxima are not unique as expected.}
        \label{fig:more-ex}
    \end{subfigure}
    \caption{A comparison of \MC and \LMC for different cuts. Red lines indicate an edge is cut by the assignment.}
    \label{fig:LMC_vs_mc_graphs}
\end{figure}

\section{Preliminaries}
For positive integer $n$, let $[n] = \{1, 2, \dots, n\}$. We consider simple, undirected $d$-regular graphs $G = (V,E)$ with $V = [n]$ and $m =\abs{E} = poly(n)$. The vertex neighborhood is denoted by the ordered set $B(v) = (u, u_1, \dots, u_d)$. For sets $A,B \subseteq [n]$, we let $A \triangle B = (A \cap B) \setminus (A \cup B)$ be the symmetric difference. Then, using the associativity of the symmetric difference we extend to a family of subsets over \([n]\), $\mathcal{F} = \left\{F_1, \dotsc, F_k\right\} \subseteq \mathcal{P}([n])$ in the following way (where \(\mathcal{P}([n])\) denotes the power set of \([n]\)). 
\[\triangle \mathcal{F} = F_1 \triangle \cdots \triangle F_k = \left\{x \in \bigcup_{i = 1}^k F_i\ |\ \#\set{F \in \mathcal{F} : x \in F} \text{ is odd}
\right\} \]

The \textit{Pauli matrices} are $2 \times 2$ unitary operators that, along with the identity matrix, $I$, form a real basis for hermitian operators a qubit:

\[I = \bmat{1 & 0 \\ 0 & 1},\ X = \bmat{0 & 1 \\ 1 & 0},\ Y = \bmat{0 & -i \\ i & 0},\ Z = \bmat{1 & 0 \\ 0 & -1}\]

We also make frequent use of the computational basis $\{ \K 0, \K 1 \}$ and the Fourier basis $\{\K +, \K - \}$ over $\H_2$, which denote the vectors

\begin{align*}
  \K{0} &= \bmat{1 \\ 0},\ \K{1} = \bmat{0 \\ 1},\\ 
  \K{+} &= \frac{1}{\sqrt{2}} = \bmat{1 \\ 1},\ \K{-} = \frac{1}{\sqrt{2}} = \bmat{1 \\ -1}  
\end{align*}

\noindent The \textit{uniform superposition} $\ket{s}$ is the state $$\ket{s} = \ket{+}^{\otimes n} = \frac1{2^{n/2}}\sum_{z \in \bitn} \ket{z}$$ For a Hermitian operator \(A\) and angle \(\theta\), the operator \[U(A,\theta) = e^{-i \theta A} \label{U-def}\] is unitary and diagonalizable over the same basis as \(A\). If $A^2 = I$, then $U(A,\theta) = \cos(\theta)I + i \sin(\theta)A$. In particular, all Pauli operators satisfy this property.

For a 1-qubit linear operator $M$, let \(M_k\) represent the corresponding operation on $\H_2^{\otimes n}$ acting as $M$ on the \(k\)th qubit and identity for the rest:

\begin{equation}
M_k = I^{\otimes k-1} \otimes M \otimes I^{\otimes n-k-1}\end{equation}

\noindent This is naturally extended to any subset $K \subseteq [n]$ as

\begin{equation}
M_K = \prod_{k \in K} M_k
\end{equation}

\noindent Note that \(M_{\emptyset} = I^{\otimes n}\). For any $S,T \subseteq [n]$ such that $S \cap T = \emptyset$  and 1-qubit operators $M$ and $N$, the operators $M_S$ and $N_T$ always commute.

\subsection{Boolean functions as Hamiltonians}
Let $C: \bitn \rightarrow \R$ be some objective function on $n$ variables which we refer to as \textit{spins}. We are interested in finding $C^* = \max_x C(x)$ as well as the input $x$ that achieves this value. An algorithm that outputs value $ALG(C)$ is an \(\alpha\)-approximation if \[\frac{ALG(C)}{C^*} \geq \alpha\] 

\noindent Typically the objective function $C$ is described by $m = poly(n)$ clauses $C_1, \dots, C_m : \bitn \rightarrow \bit$ and weights $w_1,\dots,w_m \in \R$ such that \[C(x) = \sum_{a \in [m]} w_a C_a(x)\] Spins $u$ and $v$ are \textit{neighbors} if they both appear non-trivially in at least one clause $C_a$. Using the Fourier coefficients of $C_a$, denoted by \(\hat{C}_a\), and the fact that the Pauli-$Z$ operators acting on subsets of \([n]\) give the parity functions over the computational basis, we can encode each \(C_{a}\) into a \(2^n\)-dimensional Hamiltonian operator as

\begin{equation}
  H_{C_a} = \sum_{S \subseteq [n]}\hat{C}_a(S)Z_S
\end{equation}

\noindent It is typical that a clause depends on only $k = O(1)$ many of the $n$, so we can write 

\begin{equation}
  H_{C_a} = \sum_{\substack{S \subseteq [n]\\\abs{S} \leq k}}\hat{C}_a(S)Z_S \label{bool-to-ham-clause}
\end{equation}

\noindent where $k$ is the problem's \textit{locality}. We can then sum up \eqref{bool-to-ham-clause} to find the full \textit{problem Hamiltonian}

\begin{equation}
  H_C = \sum_{\substack{S \subseteq [n]\\\abs{S} \leq k}}W_S Z_S \label{bool-to-ham}
\end{equation}

\noindent where $W_S = \sum_{a \in [m]} \hat{C}_a(S)$. This is one way the algorithm reduces from $O(2^n)$ operations to something on the scale of $poly(n)$ (with a possibly $exp(k)$ hit). See Hadfield for a full description of encoding Boolean functions into Hamiltonian operators \cite{boolean-fns-as-hams}.

\subsection{Local Algorithms}
Here we give a description of what it means for an algorithm to be local based on Hastings' tensor algorithms framework \cite{hastings2019classical}. We are given as input some interaction graph with vertices $V = [n]$ and a linear objective function $C : \bitz^n \rightarrow \R$ to optimize. As an example, the \MC objective can be written as \(\sum_{ij \in E}\frac{1-z_iz_j}{2}\). Local algorithms are in general iterative over timesteps $t = 0,1,\dots, T$ and throughout the algorithm we carry a vector $v_j^t$ for each vertex $j$ and timestep $t$. In order to calculate $(v_j^{t+1})$, we need local information about vertex $j$: $v_j^t$ and $v_k^t$ for any neighbor $k$ of $j$. After the last timestep, we use $v_j^T$ to construct an assignment $\tau : V \rightarrow \bitz^n$ such that $C(\tau(V))$ is large. In this manner, the full algorithm runs in time \(O(poly(n^k,k,T))\) which is efficient for constant \(k\).

The individual \(v_j^t\) can be taken from quite generic domains but common examples are are $\bitz$, probability distributions, or quantum states. It is assumed that matching coordinates across spins are taken from the same domain. To begin, \(v_j^0\) is assigned randomly and independently for each spin $j$. At each step $t$, \(v_j^{t}\) is constructed in two updates. First, we apply a linear update that is taken from the parts of the objective function corresponding to spin $j$ resulting in temporary vector $u_j^t$. This step depends on $v_j^t$ as well as $v_k^t$ for any neighbor $k$. Next, there is some function $g_t$ acting such that $v_k^{t+1} = g_t(u_k^{t})$. This function $g_t$ can be non-linear and random. This paper is focused on one-round local algorithms so $T=1$.

\subsection{QAOA}
The quantum approximate optimization algorithm (QAOA) \cite{FGG} is an algorithmic paradigm that attempts to solve local computational problems using low-depth quantum circuits. The algorithm is an example of a tensor algorithm that uses quantum information, which is in general iterative but we focus on the single-round instance in this work. Let $C: \bitn \rightarrow \R$ be an objective function as described in the previous section and let $H_C$ be the problem Hamiltonian the encodes $C$, as in \eqref{bool-to-ham}. We additionally note that measuring any quantum state $\ket{\psi}$ in the computational basis outputs $x \in \bitn$ with probability \(\abs{\ip{x}{\psi}}^2\). The average value of $C(x)$ over this distribution is then the expectation over this state:

\begin{equation}
  \E\limits_{x \sim \psi}[C(x)] = \mel{\psi}{H_C}{\psi}
\end{equation}

\noindent The goal of the QAOA is to find a $\ket{\psi}$ such that this expectation is large.

The QAOA is parameterized by a pair of angles \((\g, \bt) \in [0,2\pi) \times [0, \pi)\). Along with $H_C$, we also make use of a \textit{mixing operator} \(H_M\).  The framework supports a variety of mixing operators (see \cite{qaoa-mixers} for in-depth comparisons) but the important part is that they ``spread out'' amongst the eigenspaces of $H_C$ in a non-commutative way. Typically the mixing operator is an abstraction of the quantum NOT operator $X$. We use the mixing operator $H_M = \sum_v X_v$. We then construct two quantum gates \[U_C = U(H_C,\g),\quad U_M =U(H_M,\bt)\] to prepare the state \(\ket{\gamma, \beta} = U_M U_C \ket{s}\). The problem reduces to finding \((\g,\bt) \) such that

\begin{equation}
  F(\g,\bt) = \mel{\g,\bt}{H_C}{\g,\bt} \label{F}
\end{equation}

\noindent is large. This optimization is typically done as a classical-quantum hybrid algorithm making as few queries to the circuit as possible.

\begin{figure}
  \begin{subfigure}[t]{.5\textwidth}
    \centering
    \includegraphics[width=\textwidth, keepaspectratio]{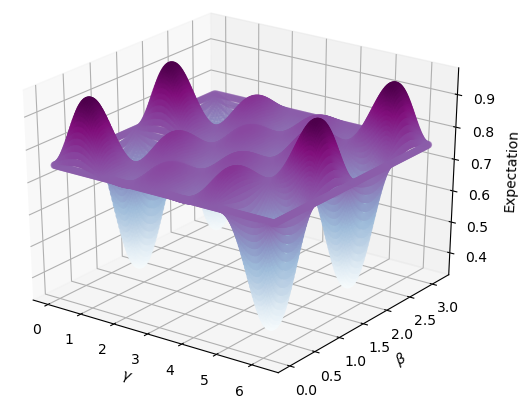}
    \caption{$d=2$}
  \end{subfigure}
  \begin{subfigure}[t]{.5\textwidth}
    \centering
    \includegraphics[width=\textwidth, keepaspectratio]{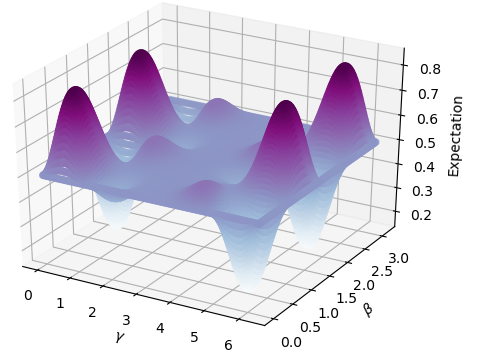}
    \caption{$d=3$}
  \end{subfigure}
  \caption{Plots of the expectation values as for all possible QAOA quantum states $\ket{\g,\bt}$ with respect to $H_2$ and $H_3$. The values are normalized by $n$ and so represent a percentage of satisfied vertices. For a fixed $\g, \bt$, we have that the expected value, $\mel{\g,\bt}{H_v^k}{\g,\bt}$, is equal to the expected value of the constraint function if we sample from the distribution induced by state $\ket{\g,\bt}$. The maxima are equal to $0.939$ in (a) and $0.819$ in (b).}
  \label{fig:quantum-3dplots}
\end{figure}

\section{Algorithms}
\subsection{Classical Algorithm}
Classically, we use a one-round local algorithm, a common family of algorithms previously studied for \MC type problems \cite{JPY,Shearer1992,threshold}. Our algorithm is straightforward and can be defined by independent actions on a vertex $v$:

\begin{itemize}
  \item[(i)] randomly assign $v$ to inital state in $\bitz$
  \item[(ii)] query $v$'s neighbors to count the number of agreeing neighbors
  \item[(iii)] update $v$'s assignment in $\bitz$ depending on the number of agreeing neighbors
\end{itemize}

More formally, a run of the algorithm is parameterized by $d + 2$ probabilities $p$ and $q  = (q_0, \dots, q_d)$. We build a cut $\tau_t$ for each timestep $t$ and define

\begin{equation}
  \ell_t(v)= \# \{uv \in E : \tau_t(u) = \tau_t(v) \} \label{agreeing}
\end{equation}

\noindent Initially, draw a random cut $\tau_0$ subject to \[\tau_0(v) =  \begin{cases}
  +1 &\text{with probability}\ p \\
  -1 &\text{otherwise}.
\end{cases}\] Next, each vertex $v$ queries a bit from its neighbors to calculate \(\ell_0(v)\). Using this value, set \[\tau_1(v) = \begin{cases} -\tau_0(v) &\text{with probability}\ q_{\ell_0(v)} \\ \tau_0(v) &\text{otherwise} \end{cases}\] The algorithm outputs cut $\tau_1$ and a vertex is then satisfied if \(\ell_1(v) \geq \ceil{\frac{d}{2}}\). This is an example of a local tensor algorithm\cite{hastings2019classical}.

Our algorithm generalizes the HRSS algorithm\cite{threshold}. Their algorithm uses threshold value \(r_d = \ceil{\frac{d + \sqrt{d}}{2}}\) to make the second step is deterministic: vertex $v$ flips if $\ell(v) \geq r_d$. In our algorithm, there is a degree of freedom for each possible $\ell(v)$ so that $v$ flips with probability $q_{\ell(v)}$. Taking for example the degree-2 case in which $r_2 = 2$, the HRSS corresponds to setting flipping probabilities $(q_0,q_1,q_2) = (0,0,1)$. That being said, there is important overlap in the low-degree cases. It turns out that in degrees 2 and 3, maximizing our algorithm over the full probability space results in very similar behavior as HRSS. The optimal strategy for our algorithm is to flip only when $\ell(v) = 2$ or $\ell(v) = 3$ for the degree-2 and degree-3 cases, respectively. However, the maximal probability might not be 1 as in HRSS. For example, in the degree-2 case, maximal probability for our algorithm occurs at $(q_0,q_1,q_2) = (0,0,4/5)$. We have the following theorem as the main result for the classical algorithm.

\begin{theorem}
    On a degree-2 graph $G$, there exists probabilities \(p,q\) such that our algorithm outputs a cut satisfying at least \(0.95n\) vertices in expectation.
\end{theorem}

\begin{theorem}
  On a degree-3 graph $G$, for all possible probabilities \(p,q\), our algorithm outputs a cut that satisfies $\leq 0.8n$ many vertices in expectation.
\end{theorem}

See Figure \ref{fig:classical-justp} for the performance of this algorithm on the low degrees. expand; mention symmetries

\begin{figure}[ht!]
  \includegraphics[width=\textwidth, keepaspectratio]{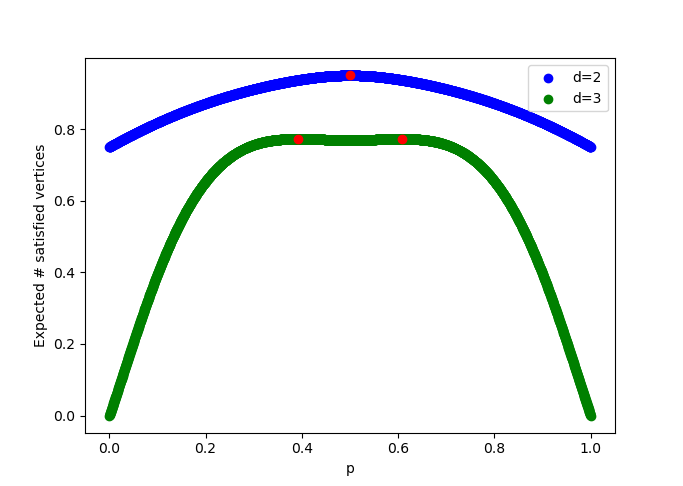}
  \caption{The optimal performance of the classical algorithm on low-degree graphs. Here, the functions for $\prob{S_v^1}$ are simple enough that calculus tricks are sufficient to isolate the maximizing variables. The x-axis ranges over probabilitiy values $p \in [0,1]$ and the y-axis is the probability a vertex is satisfied using that $p$ and the corresponding optimal $q$ values. See calculations in Appendix A for more information.}
  \label{fig:classical-justp}
\end{figure}

\subsection{QAOA Encoding}
Here we provide the encoding of the \LMC objective function into Hamiltonians as described by \eqref{bool-to-ham}. As the graph degree grows, the explicit objective function changes and so we handle the \(d=2\) and \(d=3\) cases separately.

Define the local Hamiltonian term for degree-2 graphs as

\begin{equation}
  H_v^2 = \frac{3}{4}I - \frac{1}{4}Z_{v,v_1}  - \frac{1}{4} Z_{v,v_2}  - \frac{1}{4} Z_{v_1,v_2} \label{d2-local-ham}
\end{equation} 

\noindent One can verify \ref{d2-local-ham} by checking $\mel{x}{H_v^2}{x}$ for all $x \in \bit^3$. The local terms are summed up over all vertices to build the full problem Hamiltonian

\begin{equation}
  H_2 = \sum_{v \in V} H_v^2 = \frac{3n}{4}I - \frac{1}{2}\sum_{uv \in E} Z_{u,v}- \frac{1}{4}\sum_{v \in V} Z_{v_1,v_2} \label{d2-full-ham}
\end{equation}

We now state one of the two main quantum results.

\begin{theorem}\label{quantum-d2-result-thm}
    On a degree-2 graph $G$ with large girth, every pair of angles \((\g, \bt) \) satisfies $F(\g,\bt) < 0.94n$.
\end{theorem}

This theorem is in direct contrast with Theorem 3 in which we state that the classical algorithm on degree-2 graphs can achieve at least a 0.95-approximation. Indeed, this is not too surprising. Note that the degree-2 local Hamiltonian term \ref{d2-local-ham} is 2-local, just like the \MC local constraint. So it is not surprising that the behavior of the classical versus the quantum algorithm on \LMC mimicks the behavior on \MC.

On the other hand, we start to see interesting behavior for degree-3 graphs. Let the local term $H^v_3$ and full problem Hamiltonian be given by

\begin{align}
  H_v^3 &= \frac1{2}I - \frac1{4}Z_{v,v_1} - \frac1{4}Z_{v,v_2} - \frac1{4}Z_{v,v_3} + \frac1{4}Z_{B(v)} \label{d3-local-ham} \\
  H_3 &= \sum_{v \in V}H_v^3= \frac{n}{2} - \frac1{2}\sum_{ij \in E}Z_{ij} + \frac1{4}\sum_{v \in V}Z_{B(v)} \label{d3-full-ham}
\end{align}

\begin{theorem}\label{quantum-d3-result-thm}
  On a degree-3 graph $G$ with large girth, there exist angles \((\g^*,\bt^*)\) such that $F(\g^*,\bt^*) \geq 0.81n$. 
\end{theorem}

This theorem can be compared with Theorem 4 to show that the QAOA outperforms the basic classical algorithm on degree-3 graphs with high degree. Looking at equations \eqref{d3-local-ham} and \eqref{d3-full-ham} we get a glimpse into why this might be the case. Unlike in degree-2 \LMC, we now have Pauli-$Z$ terms that rely on upwards of 4 qubits rather than just 2. We believe that this increase in complexity is crucial for allowing the QAOA to outperform the classical technique.

The proof for both theorems \ref{quantum-d2-result-thm} and the \ref{quantum-d3-result-thm} case can be found in appendix \ref{sec:quantum-proofs}.

\begin{figure}
  \begin{subfigure}[t]{.5\textwidth}
    \centering
    \includegraphics[width=\textwidth, keepaspectratio]{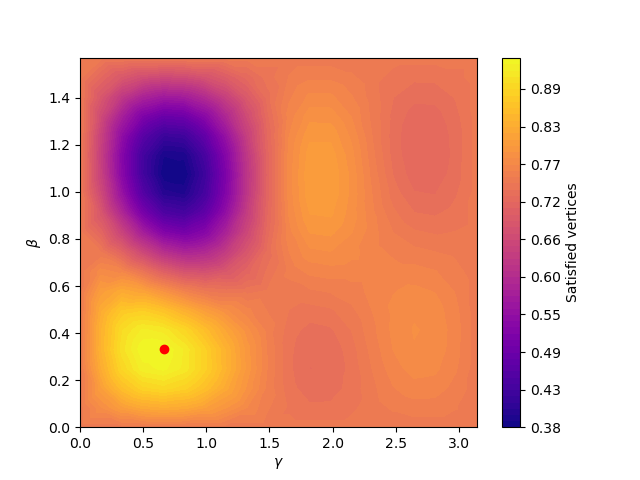}
    \caption{Degree 2}
    \label{fig:sub1}
  \end{subfigure}\hfill
  \begin{subfigure}[t]{.5\textwidth}
    \centering
    \includegraphics[width=\textwidth, keepaspectratio]{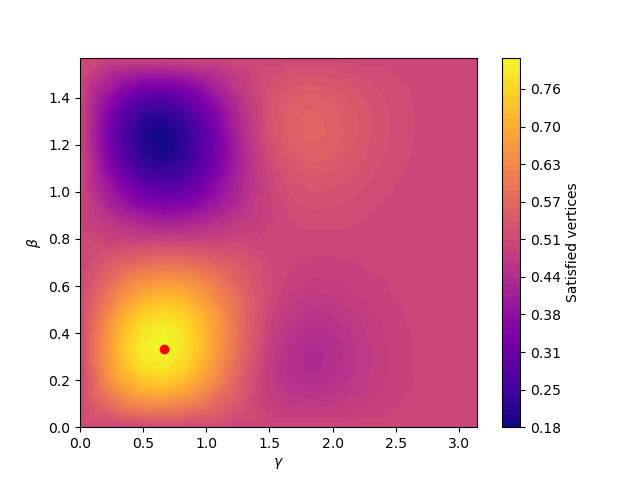}
    \caption{Degree 3}
    \label{fig:sub2}
  \end{subfigure}
  
  \caption{Plots of the expectation values for states $\ket{\g,\bt}$ with respect to $H_2$ and $H_3$ normalized by $n$. (a) and (b) display the solution surface as heatmaps on the state space \((\g,\bt) \in [0,2\pi) \times [0,\pi)\). The red dot is the maximum expectation values are approximately 0.939 and 0.819 for degree 2 and 3 graphs, respectively.}
  \label{fig:quantum-heatmaps}
\end{figure}

\bibliographystyle{unsrt}
\bibliography{local-maxcut-paper}

\appendix
\section{Classical Proofs}  
Define the following few probability events. For every vertex $v$, let $S^i_v$ be the event that $v$ is satisfied by $\tau_i$. Also let $F_v = [\tau_1(v) \neq \tau_0(v)]$ be the probability that $v$ flips its assignment between $\tau_0$ and $\tau_1$.

\begin{lemma}
    For a $d$-regular graph $G$ and initial probability $p = 1/2$, we have that $\prob{S_v^0} = \frac1{2} + \frac1{2^{d+1}}\binom{d}{d/2} = \frac1{2} + o(1)$.
\end{lemma}

The $o(1)$ term is zero for odd $d$ and is $\frac1{2^{d+1}}\binom{d}{d/2}$ for even $d$, which arises from allowing for ties.

\begin{proof} Every initial assignment of $d+1$ vertices occurs with uniform probability $\frac1{2^{d+1}}$ so this reduces to counting the number of satisfying assignments. A vertex $v$ is satisfied under $\tau_0$ when $\ell(v) \in \set{0, \dots, \floor{\frac{d}{2}}}$ of which there are $\binom{d}{j}$ many ways for each $\ell(v) = j$. Therefore,

\begin{equation}
  \label{abc}
  \prob{S_v^0} = \frac1{2^{d+1}}\sum_{j=0}^{\floor{d/2}}\binom{d}{j}
\end{equation}

Using the fact that $2^d = \sum_{j=0}^d \binom{d}{j}$ allows us to rearrange \eqref{abc} to achieve our result.
\end{proof}
  
A helpful observation is that once we have fixed a cut $\tau_0$, the probabilities of different vertices flipping are independent of one another.
  
\begin{lemma}[Independence lemma] 
    For vertices $u,v \in V$, we have that $$\prob{F_v \cap F_u | \tau_0(u), \tau_0(v)} = \prob{F_v | \tau_0(u), \tau_0(v)}\prob{F_u | \tau_0(u), \tau_0(v)}$$ This extends to any number of vertices such that $$\prob{\cap_{u \in U} F_u | \cap_{u \in U} \tau_0(u)} = \prod_{u \in U} \prob{F_u | \cap_{u \in U} \tau_0(u)}$$Moreover, if $u$ and $v$ are not neighbors of one another, then $$\prob{F_v | \tau_0(u), \tau_0(v)} = \prob{F_v | \tau_0(v)}$$
  \end{lemma}
  
  \subsection{Degree-2 graphs}
  Fix a degree-2 graph $G$ of size at least $7$. 

  \begin{lemma}
    For probabilities $p$ and $q=(q_0,q_1,q_2)$, we have that 
    
    \begin{equation}
      \begin{split}
        \prob{S_v^1 | p,q} = 1 &- (1 - p)^3 (1 - q_2) (1 - p q_1 - (1 - p) q_2)^2 - (1 - p)^3 q_2 (p q_1 + (1 - p) q_2)^2 \\
        &- p^3 (1 - q_2) (1 - (1 - p) q_1 - p q_2)^2 - p^3 q_2 ((1 - p) q_1 + p q_2)^2 \label{d2-classical-exact}
      \end{split}
    \end{equation}
    
    This function is maximized by $p=1/2$ and $q = (0,0,4/5)$ to value $0.95$.
  \end{lemma}
  
  The maximizer is found analytically using multivariable calculus techniques. There are a few simplifications we make to make the calculation simpler as well as eliminate some variables. In particular, we see that \(\prob{S_v^1 | p,q}\) does not depend on $q_0$. The first simplification we make is that in the basic case of only degree-2 graphs, a vertex that starts satisfied must remain so.
  
  \begin{lemma}
    For a degree-2 graph, if a vertex $v$ is satisfied under $\tau_0$, then it will remain satisfied under $\tau_1$. Moreover, if $v$ is satisfied under $\tau_0$, then at least one of $v$'s neighbors is also satisfied under $\tau_0$ and so will remain satisfied under $\tau_1$.
  \end{lemma}
  
  \begin{proof} Let $v_\ell$ and $v_r$ be $v$'s left and right neighbors, respectively. Assume that $v$ is satisfed under $\tau_0$ and without loss of generality, let $\tau_0(v) = 0$. There are three possible assignments for these three vertices: $$\tau_0(v_\ell,v,v_r) \in \{100, 001, 101\}$$ In the first case, both $v$ and $v_\ell$ are satisfied. Satisfied vertices never flip so $$\tau_1(v) = \tau_0(v) \neq \tau_0(v_\ell) = \tau_1(v_\ell)$$ which implies that both vertices remain satisfied under $\tau_1$. The second case is symmetrical, with $v_r$ guaranteed to be the satisfied neighbor. In the last case both neighbors are satisfied (and remain so) by equivalent reasoning.
  \end{proof}
  
  Since we are calculating $\prob{\overline{S_v^1}}$, a consequence of this lemma is that we only need to consider unsatisfied initial assignments and so  

  \begin{equation}
    \prob{\overline{S_v^1}} = \sum_{a_0 = 0}^1 \sum_{a_1 = 0}^1 \prob{\tau_1(v,v_1,v_2) = a_1 | \tau_0(v,v_1,v_2) = a_0}\prob{\tau_0(v,v_1,v_2) = a_0} \label{probd2-sum}
  \end{equation}

  For edge $uv \in E$, define conditional probabilities

  \begin{align}
    f_{00} &=\prob{F_u | \tau_0(u,v) = 00} \label{f00-2}\\
    f_{11} &=\prob{F_u | \tau_0(u,v) = 11} \label{f11-2}
  \end{align}

  \noindent It is easy to check that $$f_{00} = (1-p)q_2 + pq_1,  f_{11} = pq_2 + (1-p)q_1$$Using the independence lemma, we can calculate the conditional probabilities in \eqref{probd2-sum}:

  \begin{itemize}
    \item $\prob{\tau_1(v,v_1,v_2) = 0 | \tau_0(v,v_1,v_2) = 0} = \prob{\overline{F_v}}\cdot \overline{f_{00}}^2 = (1-q_2)\left( 1 - ((1-p)q_2 + pq_1)\right)^2$ 
    \item $\prob{\tau_1(v,v_1,v_2) = 1 | \tau_0(v,v_1,v_2) = 0} = \prob{F_v} \cdot f_{00}^2 = q_2\left( (1-p)q_2 + pq_1\right)^2$ 

    \item $\prob{\tau_1(v,v_1,v_2) = 0 | \tau_0(v,v_1,v_2) = 1} = \prob{F_v} \cdot f_{11}^2 = q_2\left(pq_2 + (1-p)q_1\right)^2$ 
    \item $\prob{\tau_1(v,v_1,v_2) = 1 | \tau_0(v,v_1,v_2) = 1} = \prob{\overline{F_v}} \cdot\overline{f_{11}}^2 = (1-q_2)\left( 1 - (pq_2 + (1-p)q_1)\right)^2$ 
  \end{itemize}
  
  Using $\prob{\tau_0(v,v_1,v_2) = 0} = (1-p)^3$ and $\prob{\tau_0(v,v_1,v_2) = 1} = p^3$ produces equation \eqref{d2-classical-exact}.

  To maximize this function, the first step is to solve \[\frac{\partial}{\partial q_2}\prob{S_v^1 | p,q_1,q_2} = 0\] which leads to the maximizer

  \begin{equation}
    q_2^*(p, q_1) = \frac{-3 + 11 p - 15 p^2 + 8 p^3 - 4 p^4 + 4 p q_1 - 14 p^2 q_1 + 20 p^3 q_1 - 10 p^4 q_1}{-6 + 26 p - 44 p^2 + 36 p^3 - 18 p^4} \label{q_2max}
  \end{equation}
  
  Plugging this into \eqref{d2-classical-exact} results in the simplification

  \begin{align}
    \prob{S_v^1 | p,q_1,q_2^*} = \frac{9 + p(-30 + p(p_1-p_2q_1(1-q_1))}{p_3}
  \end{align}

  where 

  \begin{equation}
    \begin{split}
      p_1 &= 19 + 42 p - 55 p^2 - 4 p^3 + 76 p^4 - 64 p^5 + 16 p^6 \\
      p_2 &= 4 - 36 p + 152 p^2 - 340 p^3 + 412 p^4 - 256 p^5 + 64 p^6 \\
      p_3 &= 12 - 52 p + 88 p^2 - 72 p^3 + 36 p^4
    \end{split}
  \end{equation}

  \noindent are three functions that depend only on $p$. This form is helpful because for any $q_1 \in [0,1]$, we have that \[\prob{S_v^1 | p,q_1,q_2^*} = \frac{9 + p(-30 + p(p_1-p_2q_1(1-q_1))}{p_3} \leq \frac{9 -30p+ p_1\cdot p^2}{p_3} = \prob{S_v^1 | p,0,q_2^*} \] since $q_1(1-q_1) \in [0,1/4]$. This implies that we need only consider the case when $q_1 = 0$, eliminating an additional variable. What is left is to maximize the one-variable function \[\prob{S_v^1 | p,0,q_2^*} = \frac{9 - 30 p + 19 p^2 + 42 p^3 - 55 p^4 - 4 p^5 + 76 p^6 - 64 p^7 + 
  16 p^8}{12 - 52 p + 88 p^2 - 72 p^3 + 36 p^4} \] resulting in $19/20 = 0.95$ by $p=1/2$. See Figure \ref{fig:classical-justp}.

\subsection{Degree-3 graphs}  
  Fix a degree-3 graph $G$ that is locally tree-like. The overall strategy of this section follows the previous and we would like to calculate the equation 
 
   \begin{equation}
     \prob{S_v^1|p,q} = \sum_{\nu_0 \in \bit^4} \prob{S_v^1 | \tau_0(B(v)) = \nu_0;p,q}\prob{\tau_0(B(v)) = \nu_0 |p} \label{sum-d3}
   \end{equation} 
   
   \noindent However, \eqref{sum-d3} is much more complex than the simple degree-2 case in \eqref{d2-classical-exact}.

   \begin{lemma}
    For probabilities $p$ and $q=(q_0,q_1,q_2,q_3)$, we have that \(\prob{S_v^1 | p,q}\) is maximized by $p\approx 0.39$ and $q = (0,0,0,1)$ to value $\approx 0.77$.
  \end{lemma}
  
  It is worth pointing out that we do \textit{not} want uniform initial assignment probability but actually $p \approx 3/5$. Here, it is advantageous to have a slightly worse initial cut that we can improve upon in our algorithm.

  There are many symmetries we may use to cut down on these cases as well as make each one simpler. First, we define some helpful conditional probabilities.

  \begin{lemma}
    If we define similar conditional probabilities as \eqref{f00-2} and \eqref{f11-2}, we have that

    \begin{align}
      f_{00}(p,q) &=\prob{F_v| \tau_0(uv) = 00;p,q} = q_3 (1 - p)^2 + 2 q_2p (1 - p) + q_1p^2  \label{f00-3}\\
      f_{01}(p,q) &=\prob{F_v | \tau_0(uv) = 01;p,q}= q_0(1 - p)^2 + 2 q_1p (1 - p) + q_2p^2 \label{f01-3}\\
      f_{10}(p,q) &=\prob{F_v | \tau_0(uv) = 10;p,q}= q_2(1 - p)^2 + 2 q_1 p (1 - p) + q_0 p^2 \label{f10-3}\\
      f_{11}(p,q) &=\prob{F_v | \tau_0(uv) = 11;p,q}= q_1(1 - p)^2 + 2 q_2 p (1 - p) + q_3 p^2 \label{f11-3}
    \end{align}
  \end{lemma}

  \begin{proof}
  Let $v', v''$ be the other neighbors of $v$. Consider $f_{00} = \prob{F_v | \tau_0(uv) = 00}$. There are four possible assignments $\tau_0(uvv'v'')$ to consider: $0000, 0001, 0010, 0011$. If $\tau_0(uvv'v'') = 0000$, then $v$ flips with probability $q_3$ (since $v$ agrees with 3 of its neighbors). This case occurs with probability $(1-p)^2$. When $\tau_0(uvv'v'') = 0001$ or $0010$, then $v$ flips with probability $q_2$. Each of these occur with probability $p(1-p)$. Lastly, if $\tau_0(uvv'v'') = 0011$, then $v$ flips with probability $q_1$. This case occurs with probability $p^2$. Summing these together, we get $$f_{00} = q_3 (1 - p)^2 + 2 q_2p (1 - p) + q_1p^2$$The other 3 calculations follow this same pattern.
  \end{proof}

  For bits $a,b,y$, we we want to define a function that is equal to $f_{ab}$ when $b,y$ agree and $\overline{f_{ab}}$ when they disagree. That is, 
    
  \begin{equation}
    f_{ab}^y = \begin{cases}
      f_{ab} & b \neq y\\
      \overline{f_{ab}} & b = y
    \end{cases}
  \end{equation}

  For $xyzw, abcd \in \bit^4$, we can now use the independence lemma to write

  \begin{equation}
    \prob{\tau_1(B(v)) = xyzw | \tau_0(B(v)) = abcd} = \prob{F_v = a \oplus x| \tau_0(B(v)) = abcd,f_{ab}^y f_{ac}^z f_{ad}^w }\label{p1}
  \end{equation}

  We further break down this calculation. First, using some boolean algebra and that $\overline{f_{ab}}(p,q) = 1 - f_{ab}(p,q)$, we have that this can be rewritten as

  \begin{equation}
    f_{ab}^y(p,q) =  \frac{1 - (-1)^{b \oplus y}}{2} +  (-1)^{b \oplus y}f_{ab}(p,q) \label{faby-nice}
  \end{equation}

    \begin{lemma}
        The conditional probabilities obey the following symmetries:

    \begin{enumerate}
      \item $f_{\overline{xy}}(1-p,q) = f_{xy}(p,q)$ and $f_{\overline{ab}}^{\overline{y}}(1-p,q) = f_{ab}^y(p,q)$

      \item For $q = (q_0, q_1, q_2, q_3)$, let $q^* = (q_3, q_2, q_1, q_0)$ be the vector of flipped probabilities. Then $$f_{xy}(p,q^*) = f_{x\overline{y}}(p,q), \quad f_{xy}^y(p,q^*) = f_{x\overline{y}}^y(p,q)$$
    \end{enumerate}
    \end{lemma}

    \begin{proof}
    (1) The first equality can be checked by swapping $p \mapsto 1-p$ in \eqref{f00-3} - \eqref{f11-3} and matching up corresponding equations. Then
    
    \begin{align}
      f_{\overline{ab}}^{\overline{y}}(1-p,q) &= \frac{1 - (-1)^{\overline{b} \oplus \overline{y}}}{2} +  (-1)^{\overline{b} \oplus \overline{y}}f_{\overline{ab}}(1-p,q) \\
      &= \frac{1 - (-1)^{b \oplus y}}{2} +  (-1)^{b \oplus y}f_{\overline{ab}}(1-p,q) \\
      &= \frac{1 - (-1)^{b \oplus y}}{2} +  (-1)^{b \oplus y}f_{ab}(p,q) \\
      &= f_{ab}^y(p,q) 
    \end{align}

    (2) \begin{align}
      f_{01}(p,q^*) &= q_0^*(1-p)^2 + 2q_1^*p(1-p) + q_2^*p^2 \\
      &= q_3(1-p)^2 + 2q_2p(1-p) + q_1p^2  \\
      &=p_{00}(p,q) \\
      f_{11}(p,q^*) &= q_1^*(1-p)^2 + 2q_2^*p(1-p) + q_3^*p^2 \\
      &= q_2(1-p)^2 + 2q_1p(1-p) + q_0 p^2 \\
      &= f_{10}(p,q)
    \end{align}

    We can use \eqref{faby-nice} to pass this property through to get $f_{xy}^y(p,q^*) = f_{x\overline{y}}^y(p,q)$.
    \end{proof}

    These facts are helpful to eliminate cases we need to calculate for \eqref{sum-d3}.

    \begin{lemma}
    for $xyzw, abcd \in \bit^4$, we have

    \begin{equation}
      \prob{\tau_1(B(v)) = xyzw | \tau_0(B(v)) = abcd; p,q} = \prob{\tau_1(B(v)) = \overline{xyzw} | \tau_0(B(v)) = \overline{abcd}; 1-p,q}
    \end{equation}
    \end{lemma}

    \begin{proof}Using \eqref{p1} and the previous lemma, we have that 
    \begin{align}
      &\prob{\tau_1(B(v)) = \overline{xyzw} | \tau_0(B(v))= \overline{abcd}; 1-p,q} \nonumber \\
      &\quad  = \prob{F_v = \overline{a} \oplus \overline{x}| \tau_0(B(v)) = \overline{abcd}; 1-p,q}f_{\overline{ab}}^{\overline{y}}(1-p,q) f_{\overline{ac}}^{\overline{z}}(1-p,q) f_{\overline{ad}}^{\overline{w}}(1-p,q) \nonumber\\
      &\quad  = \prob{F_v = a \oplus x| \tau_0(B(v)) = abcd; p,q}f_{ab}^{y}(p,q) f_{ac}^{z}(p,q) f_{ad}^{w}(p,q) \nonumber \\
      &\quad =\prob{\tau_1(B(v)) = xyzw | \tau_0(B(v)) = abcd; p,q}
    \end{align}
    
    \end{proof}

    This allows us to cut the number of cases in \eqref{sum-d3} in half.

    \begin{lemma}
    For any $abcd \in \bit^4$, we have that $$\prob{S_v^1 | \tau_0(B(v)) = abcd;p,q} = \prob{S_v^1 | \tau_0(B(v)) = \overline{abcd};1-p,q}$$
    \end{lemma}

    \begin{proof}
        Let $xyzw \in \bit^4$ be a satisfying assignment. Then $\overline{xyzw}$ is also a satisfying assignment.  By the previous lemma, $$\prob{\tau_1(B(v)) = xyzw | \tau_0(B(v)) = abcd; p, q} = \prob{\tau_1(B(v)) = \overline{xyzw} | \tau_0(B(v)) = \overline{abcd}; 1-p, q}$$
    \end{proof}

  Another observation is that vertex $v$ makes its decision based on its neighbor's assignments but the order does not matter. That is, for any $a,b,c,d \in \bit$,
  \begin{align}
    \prob{S_v^1 | \tau_0(B(v)) = abcd,p,q} = \prob{S_v^1 | \tau_0(B(v)) = abdc,p,q} = \prob{S_v^1 | \tau_0(B(v)) = acdb,p,q}
  \end{align} 

  This means that 
  \begin{align}
    \prob{S_v^1 | \tau_0(B(v)) = a001,p,q} = \prob{S_v^1 | \tau_0(B(v)) = a010,p,q} = \prob{S_v^1 | \tau_0(B(v)) = a100,p,q}  \\
    \prob{S_v^1| \tau_0(B(v)) = a011,p,q} = \prob{S_v^1 | \tau_0(B(v)) = a101,p,q} = \prob{S_v^1 | \tau_0(B(v)) = a110,p,q} 
  \end{align}

  Therefore, the full calculation breaks up into the following cases, where we use the shorthand $\prob{A | abcd;p,q}$ as shorthand for $\prob{A | \tau_0(B(v))=abcd;p,q}$.

  \begin{equation}\label{d3-classical-explicit}
    \begin{split}
      \prob{S_v^1} = &\prob{S_v^1 | 0000;p,q}\prob{ 0000;p } + \prob{S_v^1 |  0000;1-p,q}\prob{ 0000;1-p } \\
      &+3\left( \prob{S_v^1 |  0001;p,q}\prob{ 0001;p } + \prob{S_v^1 |  0001;1-p,q}\prob{ 0001;1-p } \right)\\
      &+3\left( \prob{S_v^1 |  0011;p,q}\prob{ 0011;p } + \prob{S_v^1 |  0011;1-p,q}\prob{ 0011;1-p } \right)\\
      &+\prob{S_v^1 |  1111;p,q}\prob{ 1111;p } + \prob{S_v^1 |  1111;1-p,q}\prob{ 1111;1-p } 
    \end{split}
  \end{equation}

  Though \eqref{d3-classical-explicit} contains many less cases than \eqref{sum-d3}, it is still a high-degree polynomial in 5 variables and so analytically maximizing it is quite difficult. Similar to the degree-2 case, we rely on a numerical optimizer to solve for the maxima here. There are more submanifolds over which this maximum occurs. As an one maximal solution is given by $p \approx 2/5, q = (0, 0, 0, 1)$ which evaluates optimized to about $0.77$.

  \section{Quantum Proofs}\label{sec:quantum-proofs}
  The goal of this section is to provide proofs for analytical expressions of $\mel{\g,\bt}{H_2}{\g,\bt}$ and $\mel{\g,\bt}{H_3}{\g,\bt}$. Let us fix some notation. Recall the form of the general problem Hamiltonian from \eqref{bool-to-ham}

  \begin{equation}
    \label{better_general_hamiltonian}
  H_C = \sum_{a = 1}^{m} H_{C_a} = \sum_{a = 1}^{m} \sum_{\substack{S \subseteq [n]\\\abs{S} \leq k}}\hat{C}_a(S)Z_S = \sum_{\substack{S \subseteq [n]\\\abs{S} \leq k}}W_S Z_S
\end{equation}

Let \(\mathcal{M} = \{M \subseteq [n]\ |\ W_M \neq 0\}\) be the collection of sets of indices that correspond to non-zero terms in \eqref{better_general_hamiltonian}. Fix some $K \subseteq [n]$. For any $L \subseteq K$, define the following two sets using $\mathcal{M}$:
  
\begin{align}
  \mathcal{O}(L) &= \{M \in \mathcal{M}\ |\ |M \cap L| \text{ is odd}\} \label{O(L)} \\
  \mathcal{O}_K(L) &= \{ \mathcal{F} \subseteq \mathcal{O}(L) \ |\  \triangle \mathcal{F} = K \} \label{O_K(L)}
\end{align}

where $\triangle \mathcal{F} = M_1 \triangle \cdots \triangle M_\ell$ is the repeated symmetric difference over the family of sets \(\mathcal{F} = \{M_1, \dotsc, M_\ell\}\). $\mathcal{O}(L)$ is all of the sets in $\mathcal{M}$ whose intersection in $L$ is odd. The next set, $\mathcal{O}_K(L)$, is a bit more complicated. For a $\mathcal{F} \in \mathcal{O}_K(L)$, we have that each $M_i \in \mathcal{F}$ is such that $M_i \in \mathcal{O}(L)$ and the symmetric difference over all $M_i \in \mathcal{F}$ is exactly $K$. These are ultimately the terms that will remain in the calculations for $\mel{\g,\bt}{Z_K}{\g,\bt}$. The following statement builds off of lemma 3.1 in \cite{ryananderson2018quantum} and is the main tool used in our analysis.

\begin{lemma}\label{anderson_lemma_3_1_but_more}
  Let $H_C$ be represented as in \eqref{better_general_hamiltonian} and \(\ket{\gamma, \beta} = U_M U_C \ket{s}\). Then for any \(K \subseteq [n]\),

  \begin{equation}\label{anderson_lemma_3_1_but_more_eq}
  \mel{\g, \bt}{Z_K}{\g, \bt} = \sum_{L \subseteq K} i^{\abs{L}}\sin(2\beta)^{\abs{L}} \cos(2\beta)^{\abs{K}-\abs{L}} \sum_{\mathcal{F} \in \mathcal{O}_K(L)}  \alpha_{\mathcal{F}}
  \end{equation}

  where \begin{equation}
    \alpha_{\mathcal{F}} = \prod_{M \in \mathcal{F}} i\sin(-2\gamma W_M)\prod_{N \in \mathcal{O}(L)\backslash \mathcal{F}} \cos(2\gamma W_N)\label{alpha} 
  \end{equation}
\end{lemma}

It is helpful to define the following components to break up \eqref{anderson_lemma_3_1_but_more_eq} even more
    
    \begin{align}
      \nu(L) &= i^{\abs{L}}\sin(2\beta)^{\abs{L}} \cos(2\beta)^{\abs{K}-\abs{L}} \label{nu} \\
      \rho(L) &= \nu(L)\sum_{\mathcal{F} \in \mathcal{O}_K(L)} \alpha_{\mathcal{F}} \label{rho}
    \end{align}
  
  Note, for \(L = \{u\}\) being a set of cordiality one, we drop the \(\set{}\) in the \(\O(\set{u})\) for ease of notation. So we write \(\O(u)\). We also do this for \eqref{O_K(L)}, \eqref{nu}, and \eqref{rho}.
  
  \begin{proof}
     We first start by stating lemma 3.1 from \cite{ryananderson2018quantum}, which is given by the following. 
    
    \begin{lemma} [Anderson lemma 3.1 from \cite{ryananderson2018quantum}, fixed\footnote{The negative was mistakenly dropped on the imaginary in equation (3.38) while applying the binomial theorem. The effects of this mistake are inconsequential to the rest of the results in \cite{ryananderson2018quantum}.}] \label{anderson_lemma_3_1_fixed}
    For \(\ket{\gamma, \beta} = U_M U_C \ket{s}\), with \(H_C\) as defined in \eqref{better_general_hamiltonian},
    \begin{equation}\label{anderson_lemma_3_1_fixed_eq}
    \begin{split}
        &\bra{\gamma, \beta} Z_{K} \ket{\gamma, \beta} = \\
        &\bra{s} Z_{K} \left(\sum_{L \subseteq K} (-i)^{|L|} \cos(2\beta)^{|K|-|L|}\sin(2\beta)^{|L|}X_L \prod_{\substack{M \subseteq [n]\\|M\cap L| \text{ is odd}}} \exp\left(-2i\gamma W_M Z_M\right) \right) \ket{s}
    \end{split}
    \end{equation}
    \end{lemma}
  
    We use this lemma as a starting point for the proof of our lemma \ref{anderson_lemma_3_1_but_more}. We note that, many of the steps for our proof are outlined in \cite{ryananderson2018quantum}, however, they are specific to the \MC problem Hamiltonian from \cite{FGG}. Additionally, similar steps are also done in other papers for solving for the expectation \cite{PhysRevA.97.022304, boolean-fns-as-hams, Marwaha2022}; however, we generalize for any real diagonal problem Hamiltonian, \(H_C\) as defined in \eqref{bool-to-ham}, and make additional observations that allow for easier analysis.
    
    One important fact we use throughout the proof is that $X \ket{+} = \ket{+}$ which extends to \(X_K \ket{s} = \ket{s}\) for any \(K \subseteq [n]\). This allows us to get rid of the \(X_L\) term in the \(Y_L\).
    
    \begin{equation}
    \begin{split}
    &\bra{\gamma, \beta} Z_{K} \ket{\gamma, \beta} \\
    =& \sum_{L \subseteq K} \cos(2\beta)^{|K|-|L|}\sin(2\beta)^{|L|} \bra{s} Y_{L} Z_{K\backslash L} \prod_{M \in \mathcal{O}(L)} \exp\left(-2i\gamma W_M Z_M\right) \ket{s} \\
    =& \sum_{L \subseteq K} \cos(2\beta)^{|K|-|L|}\sin(2\beta)^{|L|} \bra{s} i^{|L|} Z_{K} \prod_{M \in \mathcal{O}(L)} \exp\left(-2i\gamma W_M Z_M\right) \ket{s} \\
    =& \sum_{L \subseteq K} \cos(2\beta)^{|K|-|L|}\sin(2\beta)^{|L|} \bra{s} i^{|L|} Z_{K} \prod_{M \in \mathcal{O}(L)} I \cos\left(2\gamma W_M\right) + i Z_{M} \sin\left(-2\gamma W_M\right) \ket{s} \\
    =& \sum_{L \subseteq K} \cos(2\beta)^{|K|-|L|}\sin(2\beta)^{|L|} \bra{s} i^{|L|} Z_{K} \sum_{\mathcal{F} \subseteq \mathcal{O}(L)} \prod_{M\in \mathcal{F}} i Z_M \sin\left(-2\gamma W_M\right) \prod_{N \notin \mathcal{F}} I \cos\left(2\gamma W_N\right) \ket{s} \\
    =& \sum_{L \subseteq K} \cos(2\beta)^{|K|-|L|}\sin(2\beta)^{|L|} \bra{s} i^{|L|} Z_{K} \sum_{\mathcal{F} \subseteq \mathcal{O}(L)} \underbrace{\prod_{M\in \mathcal{F}} i \sin\left(-2\gamma W_M\right) \prod_{N\in \mathcal{O}(L)\backslash \mathcal{F}} \cos\left(2\gamma W_N\right)}_{\alpha_{\mathcal{F}}} \prod_{M \in \mathcal{F}} Z_M \ket{s} \\
    =& \sum_{L \subseteq K} \cos(2\beta)^{|K|-|L|}\sin(2\beta)^{|L|} \sum_{\mathcal{F} \subseteq \mathcal{O}(L)} i^{|L|} \alpha_{\mathcal{F}} \bra{s} Z_K \prod_{M \in \mathcal{F}} Z_M \ket{s}
    \end{split}
    \end{equation}
    
    Here, we turn our attention to the \(\bra{s} Z_K \prod_{M \in \mathcal{F}} Z_M \ket{s}\) term. We can utilize the fact that for any non-empty subset \(P \subseteq [n]\) we always have that \(\bra{s} Z_P \ket{s} = 0\). So, for each \(\mathcal{F} \in \mathcal{O}(L)\) term in the summation, \(\bra{s} Z_K \prod_{M \in \mathcal{F}} Z_M \ket{s}\) is non-zero when the product of Pauli-\(Z\) matrices equals \(Z_K\), i.e., \(Z_K \prod_{M \in \mathcal{F}} Z_M = I \Rightarrow \prod_{M \in \mathcal{F}} Z_M = Z_K\). This is because, pauli-\(Z\)'s on different qubits commute and \(Z^2 = I\). In other words, we have that

    \[\bra{s} Z_K \prod_{M \in \mathcal{F}} Z_M \ket{s} = \begin{cases}
        1 & \text{if } \triangle \mathcal{F} = K \\
        0 & \text{if } \triangle \mathcal{F} \neq K
    \end{cases}\]
    
    This allows us to only consider the \(\mathcal{F} \subseteq \mathcal{O}(L)\) terms when \(\triangle \mathcal{F} = K\). Notationally, we write that as considering the terms \(\mathcal{F} \in \mathcal{O}_K(L)\). Putting it together, this gives us
    
    \begin{equation}
        \bra{\gamma, \beta} Z_{K} \ket{\gamma, \beta} = \sum_{L \subseteq K} \cos(2\beta)^{|K|-|L|}\sin(2\beta)^{|L|} \sum_{\mathcal{F} \in \mathcal{O}_K(L)} i^{|L|} \alpha_{\mathcal{F}}
    \end{equation}
    \end{proof}

As these expectations tend to be high-degree trigonometric polynomials, we freely use the shorthand $\cs(\theta) = \cos(\theta)$ and $\sn(\theta) = \sin(\theta)$. There are a few applications of the double angle formula that we use as simplifications throughout our calculations.

\begin{align}
  \cos(\g)\cos\left(\frac{\g}{2}\right) + \sin(\g)\sin\left(\frac{\g}{2}\right) &= \cos\left(\frac{\g}{2}\right)\\
  \cos(\g)\cos\left(\frac{\g}{2}\right) - \sin(\g)\sin\left(\frac{\g}{2}\right) &= \cos(\frac{3\g}{2})\\
  \cos^3(\g)\sin\left(\frac{\g}{2}\right) + \sin^3(\g)\cos\left(\frac{\g}{2}\right) &= \frac1{4}\left(3 \sin (\frac{3\g}{2})-\sin\left(\frac{5\g}{2} \right) \right) \\
  \cos^3(\g)\cos\left(\frac{\g}{2}\right)-\sin^3(\g)\sin\left(\frac{\g}{2}\right) &= \frac1{4}\left(3 \cos(\frac{3\g}{2}) + \cos(\frac{5\g}{2}) \right)
\end{align}

\subsection{Degree-2}
Fix a degree-2 graph $G$ with girth at least $7$.

\begin{theorem}
  For a degree-2 graph, the full expected value of the QAOA is

  \begin{equation}\label{d2-F-exact}
      F(\gamma, \beta) = \frac{3n}{4} + \frac{n}{32}\sn(4 \beta) \left(
        3\sn(\gamma) +4 \sn(2 \gamma) + 3 \sn(3 \gamma)\right) -\frac{n}{16} \sn^2(2 \beta) \sn(\gamma)\cs^2\left(\frac{\gamma}{2} \right) \left(   \sn(\gamma) + 4 \sn(2 \gamma) + \sn(3 \gamma) \right)
  \end{equation}
  
  This equation is numerically maximized to $F(\g,\bt) \approx 0.939n < 0.94n$. 
\end{theorem}

It is important to highlight the symmetries we have in this objective function. For any edges $e_1$ and $e_2$, we have that $\mel{\g, \bt}{Z_{e_1}}{\g, \bt} = \mel{\g, \bt}{Z_{e_2}}{\g, \bt}$. For any two vertices $u$ and $v$, we have that $\mel{\g, \bt}{Z_{u_1,u_2}}{\g, \bt} = \mel{\g, \bt}{Z_{v_1,v_2}}{\g, \bt}$. So, with out loss of generality, fix an edge $uv$ and vertex $w$. Since $\abs{E(G)} = \abs{V(G)} = n$, we have that 

\begin{equation}
F(\g,\bt) = \frac{3n}{4} - \frac{n}{2} \mel{\g,\bt}{Z_{uv}}{\g,\bt} - \frac{n}{4}\mel{\g,\bt}{Z_{w_1w_2}}{\g,\bt} \label{d2-F-simplified}
\end{equation}

Solving for the expectation, \(F(\g, \bt)\), of \eqref{d2-full-ham} thus reduces to solving $\mel{\g,\bt}{Z_{uv}}{\g,\bt}$ and $\mel{\g,\bt}{Z_{w_1w_2}}{\g,\bt}$.

\begin{lemma}
  \begin{equation}
    \mel{\g,\bt}{Z_{uv}}{\g,\bt} =-2\cs(2\bt)\sn(2\bt)\cs(\g)\sn(\g)\cs^2\left(\frac{\g}{2}\right) + 2\sn^2(2\bt)\cs(\g)\sn(\g)\cs^3\left(\frac{\g}{2}\right)\sn\left(\frac{\g}{2}\right) \label{d2-exp-edge}
  \end{equation}
\end{lemma}

\begin{proof}
  Solving this expectation is a direct application of \eqref{anderson_lemma_3_1_but_more_eq} by iterating over $L \subseteq \set{u,v}$ and calculating $\rho(L)$. We use the convention for vertex labeling about the edge $uv$ given in Figure \ref{fig:d2-graph-edge-uv}.
  
\begin{figure}[ht]
  \centering
  \begin{tikzpicture}[scale=1.5, every node/.style={circle, draw, scale=1.5, line width=1pt, minimum size=15pt, inner sep=0pt}]
      \node (1) at (1,0) {$u''$};
      \node (2) at (2,0) {$u'$};
      \node (3) at (3,0) {$u$};
      \node (4) at (4,0) {$v$};
      \node (5) at (5,0) {$v'$};
      \node (6) at (6,0) {$v''$};
      
      \draw [line width=1pt] (0.25,0) -- (1) -- (2) -- (3) -- (4) -- (5) -- (6) -- (6.75,0);
  \end{tikzpicture}
  \caption{Second neighborhood graph of edge $uv \in E$ for a degree-2 graph.}
  \label{fig:d2-graph-edge-uv}
\end{figure}
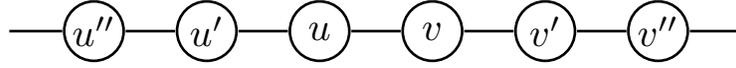

  \begin{itemize}
    \item[$L = \set{u}$: ] To begin, \[\nu(u) = i \sn(2\bt)\cs(2\bt),\  
      \mathcal{O}(\{u\}) = \set{\{u'',u\}, \{u',u\}, \{u,v\}, \{u,v'\}} \]
    
    The only element in $\O_{\set{u,v}}(u)$ is the edge $\set{u,v}$ and $\alpha_{\set{u,v}} = i\cs(\g)\sn(\g)\cs^2\left(\frac{\g}{2}\right)$. $$\rho(u) = i\sn(2\bt)\cs(2\bt)\cdot i\cs(\g)\sn(\g)\cs^2\left(\frac{\g}{2}\right) = -\sn(2\bt)\cs(2\bt)\cs(\g)\sn(\g)\cs^2\left(\frac{\g}{2}\right)$$

    \item[$L = \set{v}$: ] Due to the symmetry of this calculation we have that $\rho(v) = \rho(u)$.

    \item[$L = \set{u,v}$: ]  Here $\nu(\set{u,v}) = -\sn^2(2\bt)$ and $$\mathcal{O}(\set{u,v}) = \set{\{u'',u\}, \{u',u\}, \{u',v\}, \{u,v'\}, \{v,v'\}, \{v,v''\}}$$ Now the two elements in $\mathcal{O}_{\set{u,v}}(\set{u,v})$: $\set{\set{u', u}, \set{u', v}}$ and $\set{\set{u, v'}, \set{v, v'}}$. Both contribute the same $\alpha_{\mathcal{F}}$ values of $-\cs(\g)\sn(\g)\cs^3\left(\frac{\g}{2}\right)\sn\left(\frac{\g}{2}\right)$ so $$\rho(\set{u,v}) = 2\sn^2(2\bt)\cs(\g)\sn(\g)\cs^3\left(\frac{\g}{2}\right)\sn\left(\frac{\g}{2}\right)$$
  \end{itemize}
  
  Summing up these cases results in \eqref{d2-exp-edge}.
\end{proof}

  \begin{lemma}
    \begin{equation}
    \mel{\g,\bt}{Z_{w_1} Z_{w_2}}{\g,\bt} =-2\cs(2\bt)\sn(2\bt)\cs^2(\g)\cs\left(\frac{\g}{2}\right)\sn\left(\frac{\g}{2}\right) + \sn^2(2\bt)\cs^2(\g)\sn^2(\g)\sn^2\left(\frac{\g}{2}\right) \label{d2-exp-nonedge}
  \end{equation}
  \end{lemma}

  \begin{proof}
  This proof follows the same outline as before: sum up $\rho(L)$ for each $L \subseteq \set{w_1, w_2}$. We use the convention for vertex labeling about the vertex $w$ given in Figure \ref{fig:d2-graph-nonedge}.

\begin{figure}[ht]
  \centering
  \begin{tikzpicture}[scale=1.5, every node/.style={circle, draw, scale=1.5, line width=1pt, minimum size=15pt, inner sep=0pt}]
      \node (1) at (1,0) {$w_1''$};
      \node (2) at (2,0) {$w_1'$};
      \node (3) at (3,0) {$w_1$};
      \node (4) at (4,0) {$w$};
      \node (5) at (5,0) {$w_2$};
      \node (6) at (6,0) {$w_2'$};
      \node (7) at (7,0) {$w_2''$};
      
      \draw [line width=1pt] (0.25,0) -- (1) -- (2) -- (3) -- (4) -- (5) -- (6) -- (7) -- (7.75,0);
  \end{tikzpicture}
  \caption{Third neighborhood about vertex $w$ which corresponds to the second neighborhoods of $w_1,w_2$.}
  \label{fig:d2-graph-nonedge}
\end{figure}
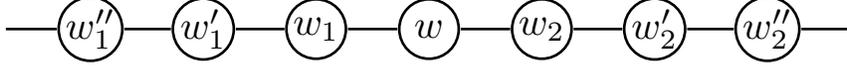

  \begin{enumerate}
    \item[$L=\set{w_1}$:] $\nu(w_1) = i \sn(2\bt)\cs(2\g)$ and $$\O(w_1) = \set{\{w_1'',w_1\}, \{w_1',w_1\}, \{w_1,w\}, \{w_1,w_2\}}$$The only solution in $\O_{\set{w_1,w_2}}(w_1)$ is the edge $\set{\set{w_1,w_2}}$, which contributes $i\cs^2(\g)\cs\left(\frac{\g}{2}\right)\sn\left(\frac{\g}{2}\right)$. Therefore $$\rho(w_1) = -1 \sn(2\bt)\cs(2\g)\cs^2(\g)\cs\left(\frac{\g}{2}\right)\sn\left(\frac{\g}{2}\right)$$

    \item[$L=\set{w_2}$:] $\rho(w_2) = \rho(w_1)$

    \item[$L=\set{w_1,w_2}$:] $\nu(\set{w_1,w_2}) = -\sn^2(2\bt)$ and $$\O(\set{w_1,w_2}) = \set{\{w_1'',w_1\}, \{w_1',w_1\}, \{w_1,w\},  \{w,w_2\}, \{w_2,w_2'\}, \{w_2,w_2''\}}$$The pair of edges $\set{\set{w_1,w},\set{w_2,w}}$ is the only solution with contribution $-\cs^2(\g)\sn^2(\g)\cs^2\left(\frac{\g}{2}\right)$. So $$\rho(\set{w_1,w_2}) = \sn^2(2\bt)\cs^2(\g)\sn^2(\g)\cs^2\left(\frac{\g}{2}\right)$$
  \end{enumerate}
  
  \end{proof}

  Plugging in the results from the previous two lemmas into \eqref{d2-F-simplified} produces equation \eqref{d2-F-exact}. Running this trigonometric polynomial through a numerical optimizer results in a maximum $F(\g, \bt)$ value of $\approx 0.93937n$, which is indeed less than $0.94n$. This proves theorem \ref{quantum-d2-result-thm}.

\subsection{Degree-3 graphs}
Fix a degree-3 graph $G$ that is locally tree-like (specifically, let it have a girth of at least 7).

\begin{theorem}
  The expected value $F$ for a degree-3 graph is given by
    \begin{equation}\label{d3-F-exact}
      \begin{split}
        F(\g, \bt) = \frac{n}{2}&+\frac{3n}{2}\cs(2 \bt) \sn(2 \bt) \sn(\g) \cs(g)\cs^4\left(\frac{\g}{2}\right)  \\
        &+\frac{n}{16} \sn(2\bt)\cs^3(2\bt)\cs^3\left(\frac{\g}{2}\right)\left(3 \sn \left(\frac{3\g}{2} \right)-\sn\left(\frac{5\g}{2} \right) \right)\\
      &+\frac{3n}{16}\sn(2\bt)\cs^3(2\bt)\sn\left(\frac{\g}{2}\right)\cs^2\left(\frac{\g}{2}\right)\left(3 \cs\left(\frac{3\g}{2}\right) + \cs\left(\frac{5\g}{2}\right) \right) \\
      &-3n\sn^3(2\bt)\cs(2\bt)\sn\left(\frac{\g}{2}\right)\cs^5(\g)\cs^5\left(\frac{\g}{2}\right)\\
      &-n\sn^3(2\bt)\cs(2\bt)\cs^6\left(\frac{\g}{2}\right)\left(\frac{1}{64} \sn\left(\frac{\g}{2}\right)\left(3 \cs\left(\frac{3\g}{2}\right)+\cs\left(\frac{5\g}{2}\right) \right)^3  + \sn^3(\g)\cs^3(\g)\cs^4\left(\frac{\g}{2}\right)\right) 
      \end{split}
    \end{equation}
    
    Moreover, there exist a pair of angles $(\g^*, \bt^*)$ such that $F(\g^*, \bt^*) \approx 0.819n > 0.81n$.
\end{theorem}

As with the degree 2 case, we note the using symmetries of \(\mel{\g, \bt}{Z_{e_1}}{\g, \bt}\) allows us to, with out loss of generality, fix an edge \(uv\) and a vertex \(u\) to express the expectation, \(F(\g, \bt)\) of \eqref{d3-full-ham}, as the following

\begin{equation}
  F(\g, \bt) = \frac{n}{2} - \frac{3n}{4}\mel{\g,\bt}{Z_{uv}}{\g,\bt} + \frac{n}{4}\mel{\g,\bt}{Z_{B(w)}}{\g,\bt} \label{d3-F-simplified}
\end{equation}

\begin{lemma}
  \begin{equation}
    \mel{\g,\bt}{Z_{uv}}{\g,\bt} = -2 \cos(2 \bt) \sin(2 \bt) \sin(\g) \cos(g)\cos^4\left(\frac{\g}{2}\right)\label{d3-exp-edge}
  \end{equation}
\end{lemma}

\begin{proof}
This proof follows the same outline as for the two in the degree 2 case. We use the convention for vertex labeling about the edge $uv$ given in Figure \ref{fig:d3-graph-edge-uv}.

\begin{figure}[ht]
    \centering
    \begin{tikzpicture}[scale=1.5, every node/.style={circle, draw, scale=1.5, line width=1pt, minimum size=15pt, inner sep=0pt}]
        \node (1) at (0,0.5) {$u'$};
        \node (2) at (0, -0.5) {$u''$};
        \node (3) at (1,0) {$u$};
        \node (4) at (2,0) {$v$};
        \node (5) at (3,0.5) {$v'$};
        \node (6) at (3,-0.5) {$v''$};
        
        \draw [line width=1pt] (1) -- (3) -- (4) -- (5) (2) -- (3) (4) -- (6);
        \draw [line width=1pt] (-0.75, 0.75) -- (1) -- (-0.75,0.25) (-0.75, -0.75) -- (2) -- (-0.75,-0.25) (3.75, 0.75) -- (5) -- (3.75,0.25) (3.75, -0.75) -- (6) -- (3.75,-0.25);
    \end{tikzpicture}
    \caption{Second neighborhood graph of edge $uv \in E$ for a degree-3 graph.}
    \label{fig:d3-graph-edge-uv}
\end{figure}
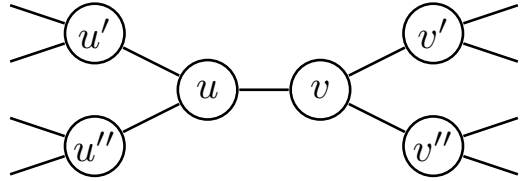

\noindent
Consider $L = \{u\}$. In this case, we have
$$\O(u) = \set{\{u,v\}, \{u,u'\}, \{u,u''\}, B(u),B(v),B(u'),B(u'')}$$

There are two solutions $\mathcal{F} \in \O(u)$ such that $\triangle \mathcal{F} = \set{u,v}$: $\mathcal{F}_1 = \set{\set{u,v}}$ and $\mathcal{F}_2 = \set{B(u), \set{u,u'}, \set{u,u''}}$. The contribution for $\mathcal{F}_1$ (and not choosing $\mathcal{F}_2$) is $i \sin(\g) \cos^2(\g)\cos\left(\frac{\g}{2}\right)$. On the other hand, if we choose $\mathcal{F}_2$ and not $\mathcal{F}_1$, the contribution is $i \sin^2(\g)\sin\left(\frac{\g}{2}\right)$. Summing these together, we have a total contribution of
$$i \sin(\g)\cos^2(g)\cos\left(\frac{\g}{2}\right) + i \sin^2(\g)\sin\left(\frac{\g}{2}\right)\cos(\g) =  i \sin(\g)\cos(\g)\cos\left(\frac{\g}{2}\right)$$
Lastly, every element in
$$\O(u) \setminus \left(\mathcal{F}_1 \cup \mathcal{F}_2\right) = \set{B(v), B(u'), B(u'')}$$
contributes $\cos\left(\frac{\g}{2}\right)$ to $\rho(u)$ since they are not used in either solution. Using $\nu(u) = i \sin(2\bt) \cos(2\bt)$, 

\begin{align}
  \rho(u)&= i \sin(2\bt) \cos(2\bt)i \cdot \sin(\g)\cos(\g)\cos\left(\frac{\g}{2}\right) \cdot \cos^3\left(\frac{\g}{2}\right)\\
  &= - \sin(2\bt) \cos(2\bt)\sin(\g)\cos(\g)\cos^4\left(\frac{\g}{2}\right)
\end{align}

Due to the symmetry of the calculation, we have $\rho(v) = \rho(u)$. The last case we need to consider is $L = \set{u,v}$. Here, $$\O(\set{u,v}) = \set{ \set{u,u'}, \set{u, u''}, \set{v, v'}, \set{v,v''}, B(u'), B(u''), B(v'), B(v'')}$$Notice that $\O_{\set{u, v}}(\set{u,v}) = \emptyset$ and so this case does not contribute any value to the expectation.
\end{proof}

\begin{lemma}
  \begin{equation}
    \begin{split}
      \mel{\g,\bt}{Z_{B(u)}}{\g,\bt} = &\frac1{4}\sn(2\bt)\cs^3(2\bt)\cs^3\left(\frac{\g}{2}\right)\left(3 \sn \left(\frac{3\g}{2} \right)-\sn\left(\frac{5\g}{2} \right) \right)\\
      &+\frac{3}{4}\sn(2\bt)\cs^3(2\bt)\sn\left(\frac{\g}{2}\right)\cs^2\left(\frac{\g}{2}\right)\left(3 \cs\left(\frac{3\g}{2}\right) + \cs\left(\frac{5\g}{2}\right) \right) \\
      &-3\sn^3(2\bt)\cs(2\bt)\sn\left(\frac{\g}{2}\right)\cs^5(\g)\cs^5\left(\frac{\g}{2}\right) \\
      &-\sn^3(2\bt)\cs(2\bt)\cs^6\left(\frac{\g}{2}\right)\left(\frac{1}{64} \sn\left(\frac{\g}{2}\right)\left(3 \cs\left(\frac{3\g}{2}\right)+\cs\left(\frac{5\g}{2}\right) \right)^3  + \sn^3(\g)\cs^3(\g)\cs^4\left(\frac{\g}{2}\right)\right)\label{d3-exp-ball}
    \end{split}
  \end{equation}
  
\end{lemma}

\begin{proof}
We define the convention for vertex labeling about the vertices \(B(u) = \set{u, u_1, u_2, u_3}\) with Figure \ref{fig:d3-graph-edge-Bu}.

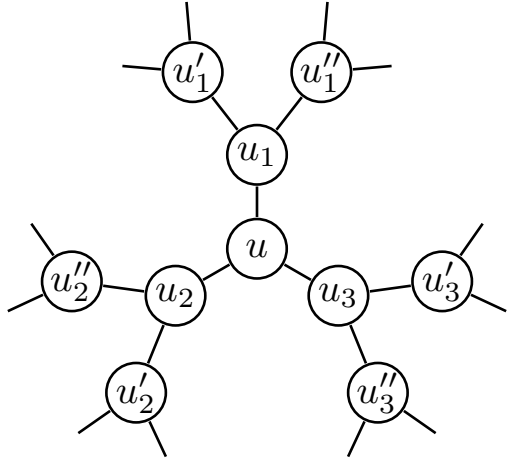
\begin{figure}[ht]
    \centering
    \begin{tikzpicture}[scale=1.25, every node/.style={circle, draw, scale=1.5, line width=1pt, minimum size=15pt, inner sep=0pt}]
        \node (u) at (0,0) {$u$};
        
        \foreach \i in {1,2,3}{
            \node (u\i) at (\i*360/3 - 30: 1) {$u_{\i}$};
            \node (up\i) at (\i*360/3 - 10: 2) {$u'_{\i}$};
            \node (upp\i) at (\i*360/3 - 50: 2) {$u''_{\i}$};
            
            \draw [line width=1pt] (u) -- (u\i);
            \draw [line width=1pt] (upp\i) -- (u\i) -- (up\i);
            
            \draw [line width=1pt] ($(up\i)+(\i*360/3 + 55: 0.75)$) -- (up\i) -- ($(up\i)+(\i*360/3 - 25: 0.75)$);
            \draw [line width=1pt] ($(upp\i)+(\i*360/3 + 55 - 90: 0.75)$) -- (upp\i) -- ($(upp\i)+(\i*360/3 - 25 - 90: 0.75)$);
        }
    \end{tikzpicture}
    \caption{Third neighborhood about vertex $u$ which corresponds to the second neighborhoods of $u_1,u_2,u_3$.}
    \label{fig:d3-graph-edge-Bu}
\end{figure}

\noindent
We work out the calculation $\rho(L)$ for each $L \subseteq \set{u,u_1,u_2,u_3}$ and them sum up to get \eqref{d3-exp-ball} as per \eqref{anderson_lemma_3_1_but_more_eq}.

\begin{enumerate}
  \item[$L = \set{u}$:] To begin, $\nu(u) = i\sn(2\bt)\cs^3(2\bt)$ and $$\O(u) = \set{\set{u,u_1}, \set{u,u_2}, \set{u,u_3}, B(u), B(u_1), B(u_2), B(u_3)}$$ There are two solutions $\mathcal{F}_1 = \{B(u)\}$ and $\mathcal{F}_2 = \set{\set{u,u_1}, \set{u,u_2}, \set{u,u_3}}$ in $\O_{B(u)}(u)$. These solutions are mutually exclusive so we can sum up their individual contributions to $\rho(u)$. Solution $\mathcal{F}_1$ contributes $-i\sn\left(\frac{\g}{2}\right)\cs^3(\g)$ and $\mathcal{F}_2$ contributes $(i\sn(\g))^3\cs\left(\frac{\g}{2}\right) = -i \sn^3(\g)\cs\left(\frac{\g}{2}\right)$. Next, notice that the elements in $\O(u) \setminus (\mathcal{F}_1 \cup \mathcal{F}_2) = \set{B(u_1), B(u_2), B(u_3)}$ are not part of either solution and so each contribute $\cs\left(\frac{\g}{2}\right)$. Putting this all together,
  
  \begin{align}
    \rho(u) &= i\sn(2\bt)\cs^3(2\bt) \left[-i\sn\left(\frac{\g}{2}\right)\cs^3(\g) - i \sn^3(\g)\cs\left(\frac{\g}{2}\right)\right]\cs^3\left(\frac{\g}{2}\right) \\
    &=  \frac1{4}\sn(2\bt)\cs^3(2\bt)\cs^3\left(\frac{\g}{2}\right)\left(3 \sn \left(\frac{3\g}{2} \right)-\sn\left(\frac{5\g}{2} \right) \right) \label{ru}
  \end{align}

  \item[$L = \set{u_1}$:] We have $\nu(u_1) = i\sn(2\bt)\cs^3(2\bt)$ and $$\O(u_1) = \set{\set{u,u_1 }, \set{u_1,u_1'}, \set{u_1,u_1''}, B(u), B(u_1), B(u_1'), B(u_1'')}$$There are two ways to get a symmetric difference of $B(u)$. Define $E_1 = \set{B(u_1), \set{u,u_1}, \set{u_1,u_1'}, \set{u_1,u_1''}}$. Then $\triangle E_1 = \emptyset$ and so both $\set{B(u)}$ and $\set{B(u)} \cup E_1$ are in $\O_{B(u)}(u_1)$. Both solutions use $B(u)$ and so have a $-i \sn\left(\frac{\g}{2}\right)$ term. If we omit $E_1$, the contribution is $\cs^3(\g)\cs\left(\frac{\g}{2}\right)$ and if we include $E_1$, the contribution is $-\sn^3(\g)\sn\left(\frac{\g}{2}\right)$. The remaining contribution comes from the elements in $\O(u_1)\setminus (\set{B(u)} \cup E_1) = \set{B(u_1'), B(u_1'')}$ which always contribute $\cs\left(\frac{\g}{2}\right)$ each. Putting these together, we have
  
  \begin{align}
    \rho(u_1) &= i\sn(2\bt)\cs^3(2\bt)\cdot (-i) \sn\left(\frac{\g}{2}\right)\left[\cs^3(\g)\cs\left(\frac{\g}{2}\right)-\sn^3(\g)\sn\left(\frac{\g}{2}\right)\right] \cdot \cs^2\left(\frac{\g}{2}\right) \\
    &= \frac1{4}\sn(2\bt)\cs^3(2\bt)\sn\left(\frac{\g}{2}\right)\cs^2\left(\frac{\g}{2}\right)\left(3 \cs(\frac{3\g}{2}) + \cs(\frac{5\g}{2}) \right) \label{ru1}
  \end{align}

  \item[$L = \set{u_2}$, $L = \set{u_3}$:] These cases are the same as $L = \set{u_1}$.
  
  \item[$L = \set{u, u_1}$:] Note that $$\O(\set{u,u_1}) = \set{\{u,u_2 \}, \{u,u_3  \}, \{u_1,u_1' \},\{u_1,u_1'' \}, B(u_2),B(u_3),B(u_1'),B(u_1'')}$$contains no subsets whose symmetric difference is equal to $B(u)$ and so $\rho(\set{u,u_1}) = 0$.

  \item[$L = \set{u, u_2}, L = \set{u, u_3}$:] These cases are the same as $L = \set{u,u_1}$ and also contribute 0.

  \item[$L = \set{u_1, u_2}$:] Note that $$\O(\set{u_1,u_2}) = \set{\{u,u_1 \}, \{u_1,u_1' \}, \{u_1,u_1'' \},\{u,u_2 \}, \{u_2,u_2' \}, \{u_2,u_2'' \}, B(u_1'), B(u_1''), B(u_2'), B(u_2'')}$$contains no subsets whose symmetric difference is equal to $B(u)$ and so $\rho(\set{u_1,u_2}) = 0$.

  \item[$L = \set{u_1, u_3}, L = \set{u_2, u_3}$:] These cases are the same as $L = \set{u_1,u_2}$ and also contribute 0.

  \item[$L = \set{u, u_1, u_2}$:] We have $\nu(\set{u, u_1, u_2}) = -i\sn^3(2\bt)\cs(2\bt)$ and $$\O(\set{u, u_1, u_2}) = \set{\{u,u_3 \}, \{u_1,u_1' \}, \{u_1,u_1'' \},\{u_2,u_2' \}, \{u_2,u_2'' \}, B(u), B(u_3),  B(u_1'), B(u_1''),  B(u_2'), B(u_2'')}$$The only solution here is $B(u)$ which contributes $-i\sn\left(\frac{\g}{2}\right)\cs^5(\g)\cs^5\left(\frac{\g}{2}\right)$. So we have 
  
  \begin{align}
    \rho(\set{u, u_1, u_2}) = -i\sn^3(2\bt)\cs(2\bt)\cdot-i\sn\left(\frac{\g}{2}\right)\cs^5(\g)\cs^5\left(\frac{\g}{2}\right) = -\sn^3(2\bt)\cs(2\bt)\sn\left(\frac{\g}{2}\right)\cs^5(\g)\cs^5\left(\frac{\g}{2}\right) \label{ruu1u2}  
  \end{align}

  \item[$L = \set{u, u_2, u_3}, L = \set{u, u_1, u_3}$:] These cases are the same as $L = \set{u, u_1, u_2}$.

  \item[$L = \set{u_1, u_2, u_3}$:] This case is the most complicated. First, note that

  \begin{align*}
    \mathcal{O}(\set{u_1, u_2, u_3}) &= \{\{u,u_1 \}, \{u,u_2 \}, \{u,u_3 \}, \{u_1,u_1' \}, \{u_1,u_1'' \}, \{u_2,u_2' \}, \{u_2,u_2'' \},\{u_3,u_3' \}, \{u_3,u_3'' \},\\ 
    &\qquad B(u), B(u_1), B(u_1'), B(u_1''), B(u_2), B(u_2'), B(u_2''),B(u_3), B(u_3'), B(u_3'')\}
  \end{align*}
  
  Similar to the $L = \set{u}$ case, any $\mathcal{F} \in \mathcal{O}_{B(u)}(\set{u_1, u_2, u_3})$ corresponds to either $\mathcal{F}_1 = \{B(u)\}$ or $\mathcal{F}_2 = \set{\set{u,u_1},\set{u,u_2},\set{u,u_3}}$. However, there are now many ways to result in these sets. Define the following sets $$E_{j} = \set{B(u_j), \set{u_j, u},\set{u_j, u_j''},\set{u_j, u_j'''}}, F_j = E_j \setminus \set{\set{u_j, u}}$$ for each $j \in \set{1,2,3}$. Then we have that
  
  \begin{align}
    \triangle E_j &= \emptyset \label{delta-ej}\\
    \triangle F_j &= \set{u, u_j} \label{delta-fj}
  \end{align}
  
  For the solution $\mathcal{F}_1$, by \eqref{delta-ej}, we can construct another solution that contain any of the $E_j$. On the other hand, for the solution $\mathcal{F}_2$ we can remove any of the edges $\set{u,u_j}$ in-place of $F_j$ to get another solution, by \eqref{delta-fj}. As in the case of $L = \set{u}$, $\mathcal{F}_1$ solutions are mutually exclusive from $\mathcal{F}_2$ solutions and so we sum up their separate calculations. We also note these 16 solutions make up all the possible solution $\mathcal{F} \in \mathcal{O}_{B(u)}(\set{u_1, u_2, u_3})$.
  
  Beginning with $\mathcal{F}_1$ solutions we can independently choose to include $E_1, E_2,$ and $E_3$ so there are $2^3 = 8$ possible solutions in this case. Every solution contains $B(u)$ which contributes $-i\sn\left(\frac{\g}{2}\right)$. Start with deciding whether to pick $E_1$. If we do not include $E_1$, this solution contributes $\cs^3(\g)\cs\left(\frac{\g}{2}\right)$. If we do include $E_1$, then we pick up $(i\sn(\g))^3(-i\sn\left(\frac{\g}{2}\right)) = - \sn^3(\g)\sn\left(\frac{\g}{2}\right)$. Summing up these cases and using a double angle formula results in $$\cs^3(\g)\cs\left(\frac{\g}{2}\right)- \sn^3(\g)\sn\left(\frac{\g}{2}\right) = \frac1{4}\left(3 \cs\left(\frac{3\g}{2}\right)+\cs\left(\frac{5\g}{2}\right) \right)$$This is also the contribution concerning $E_2$ and $E_3$ and these cases are independent. Lastly, any $$\mathcal{O}(L)\setminus \left(\{B(u)\} \cup E_1 \cup E_2 \cup E_3\right) = \set{B(u_1'), B(u_1''), B(u_2'), B(u_2''), B(u_3'), B(u_3'')}$$is not chosen contributes an additional $\cs\left(\frac{\g}{2}\right)$. Therefore, the contribution to $\rho(\set{u_1, u_2, u_3})$ using $\mathcal{F}_1$ is equal to
  
  \begin{align}
    -i\sn\left(\frac{\g}{2}\right)\left[ \frac1{4}\left(3 \cs\left(\frac{3\g}{2}\right)+\cs\left(\frac{5\g}{2}\right) \right) \right]^3 \cs^6\left(\frac{\g}{2}\right) = -\frac{1}{64}i \sn\left(\frac{\g}{2}\right)\left(3 \cs\left(\frac{3\g}{2}\right)+\cs\left(\frac{5\g}{2}\right) \right)^3 \cs^6\left(\frac{\g}{2}\right)
  \end{align}
  
  Next we need to find the contribution to $\rho(\set{u_1, u_2, u_3})$ using the $\mathcal{F}_2$ solutions. We can get $\set{u,u_1}$ by either choosing the edge itself or the set $F_1 = \set{B(u_1), \set{u_1,u_1'}, \set{u_1,u_1''}}$ since $\triangle F_1 = \set{u, u_1}$. If we chose the edge and not $F_1$, then we have a contribution of $i\sn(\g)\cs^2(\g)\cs\left(\frac{\g}{2}\right)$. If we choose $F_1$ and not the edge, this contributes $i \sn^2(\g)\sn\left(\frac{\g}{2}\right)\cs(\g)$ The full contribution for the edge $\set{u,u_1}$ is then $$i\sn(\g)\cs^2(\g)\cs\left(\frac{\g}{2}\right)+i \sn^2(\g)\sn\left(\frac{\g}{2}\right)\cs(\g) = i \sn(\g)\cs(\g)\left[\cs(\g)\cs\left(\frac{\g}{2}\right) + \sn(\g)\sn\left(\frac{\g}{2}\right)\right] = i \sn(\g)\cs(\g)\cs\left(\frac{\g}{2} \right)$$ This logic is the same for independently choosing $\set{u,u_2}$ versus $F_2$ and same for $\set{u,u_3}$ and $F_3$. Lastly, elements in $$\mathcal{O}(L)\setminus \left(\set{\set{u,u_1}, \set{u,u_2}, \set{u,u_3}} \cup F_1 \cup F_2 \cup F_3\right) = \set{B(u), B(u_1'), B(u_1''), B(u_2'), B(u_2''), B(u_3'), B(u_3'')}$$are never chosen in our solution and each contributes $\cs\left(\frac{\g}{2}\right)$. The full contribution to $\rho(\set{u_1, u_2, u_3})$ using $\mathcal{F}_2$ is then $$\left( i \sn(\g)\cs(\g)\cs\left(\frac{\g}{2} \right)\right)^3\cs^7\left(\frac{\g}{2}\right) = -i \sn^3(\g)\cs^3(\g)\cs^{10}\left(\frac{\g}{2}\right) $$We have that $\nu(\set{u_1, u_2, u_3}) = -i \sn^3(2\bt)\cs(2\bt)$ and so 
  
  \begin{align}
    \rho(\set{u_1, u_2, u_3}) &= -i \sn^3(2\bt)\cs(2\bt)\left(-\frac{1}{64}i \sn\left(\frac{\g}{2}\right)\left(3 \cs\left(\frac{3\g}{2}\right)+\cs\left(\frac{5\g}{2}\right) \right)^3 \cs^6\left(\frac{\g}{2}\right) -i \sn^3(\g)\cs^3(\g)\cs^{10}\left(\frac{\g}{2}\right)\right) \nonumber \\
    &=-\sn^3(2\bt)\cs(2\bt)\cs^6\left(\frac{\g}{2}\right)\left(\frac{1}{64} \sn\left(\frac{\g}{2}\right)\left(3 \cs\left(\frac{3\g}{2}\right)+\cs\left(\frac{5\g}{2}\right) \right)^3  + \sn^3(\g)\cs^3(\g)\cs^4\left(\frac{\g}{2}\right)\right)\label{ru1u2u3}  
  \end{align} 
\end{enumerate}

All that is left is to sum the cases to get $\eqref{d3-exp-ball} = \eqref{ru} + 3\cdot\eqref{ru1} + 3\cdot \eqref{ruu1u2} + \eqref{ru1u2u3} $
\end{proof}

Plugging in lemmas 10 and 11 into \eqref{d3-F-simplified} results in \eqref{d3-F-exact}. Running this trigonometric polynomial through a numerical optimizer results in a maximum $F(\g, \bt)$ value of $\approx 0.819292n$, which is indeed greater than \(0.81n\). This proves theorem \ref{quantum-d3-result-thm}.

\end{document}